\documentclass[envcountsect,envcountsame]{llncs} 
\usepackage{ifthen}
\newcounter{Accumulate} 
\ifthenelse{\value{Accumulate} = 1}{
  \newwrite\accuwrite \immediate\openout\accuwrite=\jobname.acc
}{}
\usepackage{environ}
\makeatletter
\newenvironment{accumulate}{\Collect@Body\accuPrint}{}
\makeatother
\newcommand{\accuPrint}[1]{
 \ifthenelse{\value{Accumulate} = 0}{%
      #1
  }
  {
    \newtoks\prxxxm
    \prxxxm{#1}
    \immediate\write\accuwrite{\the\prxxxm}
  }
}
\newcommand{\ifaccumulating}[1]{%
  \ifthenelse{\value{Accumulate} = 1}{%
    #1%
  }{}%
}
\newcommand{\ifnoaccumulating}[1]{%
  \ifthenelse{\value{Accumulate} = 0}{%
    #1%
  }{}%
}
\newcommand{\accuprint}{%
  \ifthenelse{\value{Accumulate} = 1}{
    \immediate\closeout\accuwrite
    \input{\jobname.acc} %
  }{}
}

\def\ca#1{{\cal#1}}
\def\crgj#1{\mathop{\rm cr}_{#1}}

\usepackage{amsmath,amssymb}
\usepackage{xspace}
\usepackage{color,enumerate}
\usepackage{tikz}
\tikzstyle{every picture} = [>=latex]

\usepackage{todonotes}

\usepackage[czech,english]{babel}

\title{On Hardness of the Joint Crossing Number}
\author{Petr Hlin\v{e}n\'{y}\inst{1} \and Gelasio Salazar\inst{2}}
\institute{%
Faculty of Informatics, Masaryk University Brno, Czech Republic\\
\email{hlineny@fi.muni.cz} \and
Instituto de Fisica, Universidad Autonoma de San Luis Potosi, Mexico
\email{gsalazar@ifisica.uaslp.mx}
}

\begin{document}
\maketitle

\begin{abstract}
The Joint Crossing Number problem 
asks for a simultaneous embedding of two disjoint graphs into one surface 
such that the number of edge crossings (between the two graphs) is minimized.
It was introduced by Negami in 2001 in
connection with diagonal flips in triangulations of surfaces,
and subsequently investigated in a general form for small-genus surfaces.
% In general, the problem asks for a simultaneous embedding of two disjoint
% graphs into one surface such that the number of edge crossings
% is minimized. 
We prove that all of the commonly considered variants of this problem are
NP-hard already in the orientable surface of genus~$6$, 
by a reduction from a special variant of the anchored crossing number
problem of Cabello and Mohar.
\end{abstract}

\section{Introduction}

Motivated by his investigation on diagonal flips in triangulations of
surfaces~\cite{negamiflips}, Negami introduced in~\cite{DBLP:journals/jgt/Negami01} the concept of {\em
  joint crossing numbers}. The general setup consists of two graphs
embeddable on the same surface, and the problem is to find a
simultaneous embedding into this surface, so that the number of edge
crossings is minimized. (Since both graphs are embedded, 
every crossing must involve an edge from each of the graphs.) 

In Negami's original definition, the embedded graphs were allowed to
share vertices and edges (this is the {\em diagonal crossing number}). 
In the subsequent papers on
joint crossing numbers, the attention has been restricted to the case
in which the corresponding graphs are disjoint. This mainstream
case is the one we focus on in this work, and we restrict the attention
to orientable surfaces. 
% Moreover, we investigate the case in which the host surface is orientable.

Within this case (the graphs $G_1, G_2$ to be jointly embedded in the
same surface $\Sigma$ are disjoint), three variants proposed by Negami have been studied. In the
first one, the aim is to minimize the number of crossings in 
{\em any} embedding of the disjoint union $G_1+G_2$ of $G_1$ and $G_2$; this is simply the {\em joint
  crossing number}. In the second variant, the {\em joint homeomorphic
  crossing number}, embeddings of $G_1$ and $G_2$ are already given,
and one must embed $G_1+G_2$ so that the restriction of this embedding
to each $G_i$ is homeomorphic to the prescribed embedding of $G_i$. In the third,
and most restricted variant, the {\em joint orientation-preserving
  homeomorphic crossing number}, in addition, the restrictions of the
embedding of $G_1+G_2$ to each $G_i$ must be orientation-preserving
homeomorphic to the prescribed embedding of $G_i$. 
(See the next section for more rigorous definitions.)
% (More rigorous
% definitions shall be given in the next section; for the time being we
% prefer to be slightly informal so we can state our main results).

Relatively little is known on either of these
variants. In~\cite{DBLP:journals/jgt/Negami01}, Negami bounded the
homeomorphic crossing number in terms of the Betti numbers of the
graphs and the genus of
$\Sigma$. In~\cite{DBLP:journals/jgt/ArchdeaconB01}, Archdeacon and
Bonnington calculated the exact homeomorphic crossing number of two
graphs embedded in the projective plane, and also gave lower and upper
bounds, within a constant factor of each other, for the case in which
the host surface is the torus (Negami also obtained some nontrivial bounds
for toroidal joint embeddings
in~\cite{DBLP:journals/jgt/Negami01}). Richter and Salazar
investigated in~\cite{DBLP:journals/jgt/RichterS05} the case in which
both graphs are densely embedded. 

% From the computational complexity point of view, the associated
% decision problems are the following:
The associated algorithmic problems are the following:

\medskip
\noindent{\sc Joint Crossing Number} 
\smallskip

\noindent{\bf Input:} Graphs $G_1, G_2$ embeddable in a given surface
$\Sigma$, and an integer~$k$.\\
\noindent{\bf Question:} Is the joint crossing number of $G_1$ and
$G_2$ in $\Sigma$ at most $k$?
\medskip

\noindent{\sc Joint Homeomorphic Crossing Number} 
\smallskip

\noindent{\bf Input:} Embeddings of each of two disjoint graphs $G_1$,
$G_2$ in a surface $\Sigma$, and~$k$.\\
\noindent{\bf Question:} Is the joint homeomorphic crossing number of $G_1$ and
$G_2$ at most~$k$?%
\medskip

\noindent{\sc Joint OP-Homeomorphic Crossing Number} 
\smallskip

\noindent{\bf Input:} Embeddings of each of two disjoint graphs $G_1$,
$G_2$ in a surface $\Sigma$, and~$k$.\\
% \noindent{\bf Input:} Embeddings of two disjoint graphs $G_1$
% and $G_2$ on the same surface $\Sigma$.\\
\noindent{\bf Question:} Is the joint orientation-preserving homeomorphic 
crossing number of $G_1$ and $G_2$ in $\Sigma$ at most~$k$?
\medskip

It follows from~\cite[Theorem 2.2]{DBLP:journals/jgt/ArchdeaconB01}
that the last two problem variants are easy in the projective plane
(it suffices to calculate the dual widths of the embeddings).
The aforementioned results also suggest (although this is an open problem) that
optimal solutions in the case of $\Sigma$ being the torus,
can always be obtained in a particularly nice way:
embed all the vertices of one of the graphs, say $G_1$, in the same face of
the other graph $G_2$ (and then route the excessive edges of $G_1$ across $G_2$).
This nice property ceases to be true for the homeomorphic variant already in
the double torus~\cite{DBLP:journals/jgt/RichterS05}, and our results imply
that the property fails really badly for all higher genus surfaces and all
problem variants.

In his comprehensive survey~\cite{survscha} of the many different variants of
crossing number definitions, Marcus Schaefer marks the complexity of all the
aforementioned variants of the joint crossing number as open. 
These problems are all easily seen to be in NP, so the open problem is
their hardness. Our main result in this paper settles this question.

\begin{theorem}\label{thm:maintheorem}
{\sc Joint Crossing Number}, {\sc Joint Homeomorphic Crossing Number}, and {\sc
 Joint OP-Homeomorphic Crossing Number} are NP-hard problems
in any orientable surface of genus $6$ or higher.
This remains true even if the inputs are restricted to simple $3$-connected graphs.
\end{theorem}

The proof of this theorem is via a chain of reductions from a special variant of the
anchored crossing number problem of
Cabello and Mohar~\cite{DBLP:journals/siamcomp/CabelloM13}.

The rest of this paper is organized as follows. In
Section~\ref{sec:basic} we give rigorous definitions of the variants
of the joint embedding problem we analyze, and review some basic
concepts. To work out the reduction from Cabello-Mohar's anchored
crossing number, we devise joint embeddings in which certain vertices
of one of the graphs are required to lie
% (in any crossing-minimizing joint embedding) 
in prescribed faces of the other embedding. These {\em
  face-anchored joint embeddings} are developed in
Section~\ref{sec:face-anchored}. An additional fine-tuning of the
construction (in order to bring the reduction to anchored crossing
number from graphs of bounded genus) is given in
Section~\ref{sec:mult}. The reduction to the anchored crossing number is
laid out in Section~\ref{sec:reduction}. 
Finally, in
Section~\ref{sec:conclusions} we present some concluding remarks,
among which we give back a slight strengthening of the main result of
aforementioned~\cite{DBLP:journals/siamcomp/CabelloM13}.

\section{Basic Concepts}\label{sec:basic}
%%%%%%%%%%%%%%%%%%%%%%%%%%%%%%%%%%%%%%%%%%%%%%

We follow the standard notation of letting $G_1+G_2$ denote the graph obtained as the disjoint union of two graphs
$G_1,G_2$.
A {\em toroidal grid} of size $p\times q$ is the Cartesian product of
a $p$-cycle with a $q$-cycle; this is a $4$-regular graph
consisting of an edge disjoint union of $q$
copies of a $p$-cycle and $p$ copies of a $q$-cycle. For each integer
$h\ge 0$, we let $\ca S_h$ denote the orientable surface of genus $h$.

We recall that in a {\em drawing} of a graph $G$ in a surface $\Sigma$,
vertices are mapped to points and edges are mapped to simple curves
(arcs) such that the endpoints of an arc are the vertices of the
corresponding edge; no arc contains a point that
represents a non-incident vertex. For simplicity, we often make no
distinction between the topological objects of a drawing (points and
arcs) and their corresponding graph theoretical objects (vertices and edges). 
A {\em crossing} in a drawing is an intersection point of two edges 
%%(or a self-intersection of one edge) 
in a point other than a common endvertex. 
An {\em embedding} of a graph in a surface is a drawing with no edge crossings.

\begin{accumulate}
\ifaccumulating{\subsection{More basic concepts}}
We shall make use of the following facts; they are all straightforward
exercises in topological graph theory.

\begin{proposition}
\label{pro:disj-nonplanar-a}
Assume that a disjoint union of $k$ non-planar graphs is embedded in the
surface $\ca S_h$. Then $h\geq k$.
\end{proposition}

\begin{proposition}
\label{pro:disj-nonplanar-b}
Let $G$ be a graph embedded in the surface $\ca S_h$.
Suppose that there exist $h$ pairwise disjoint non-planar subgraphs
$F_i\subseteq G$, $i=1,\dots,h$.
Then the induced embedding of $G-\bigcup_{i=1}^h V(F_i)$ is plane.
\end{proposition}

\begin{proposition}
\label{pro:k33handle}
Suppose that a graph $G$, isomorphic to $K_{3,3}$, is embedded in a
surface $\Sigma$. Let $v\in V(G)$. 
Then at most one of the three cycles of $G-v$ is contractible in~$\Sigma$.
\end{proposition}
\end{accumulate}\medskip

Let $G_1,G_2$ be disjoint graphs, both of which embed in the same
orientable surface $\Sigma$. A drawing $G^0$ of the graph $G_1+G_2$ in $\Sigma$ is called a
{\em joint embedding of $(G_1,G_2)$} 
if the restriction of $G^0$ to $G_i$, for each $i=1,2$, is an embedding.
Furthermore, if prescribed embeddings $G^0_1,G^0_2$ of $G_1,G_2$ are given
and the restriction of $G^0$ to $G_i$ is homeomorphic (respectively,
orientation-preserving homeomorphic) to $G^0_i$, $i=1,2$,
then $G^0$ is a {\em joint homeomorphic
embedding of $(G^0_1,G^0_2)$} (respectively, {\em joint
  orientation-preserving homeomorphic
embedding of $(G^0_1,G^0_2)$}) in $\Sigma$.
Loosely speaking, in the joint homeomorphic variant(s), one is only allowed
to ``deform'' the prescribed embeddings of $G_1,G_2$ across the host surface.

Note that in any joint embedding of $(G_1,G_2)$, crossings may arise only 
between an edge of $G_1$ and an edge of $G_2$.
The {\em joint crossing number} of $(G_1,G_2)$ in $\Sigma$
%denoted by $\crgj{\Sigma}(G_1,G_2)$,
is the minimum number of crossings over all joint embeddings of
$(G_1,G_2)$ in $\Sigma$.
The {\em joint homeomorphic crossing number} and {\em joint
orientation-preserving homeomorphic crossing number} are defined
analogously.

In order to resolve the ordinary and homeomorphic variants of joint crossing
number problems at once, we introduce the following ``generalizing''
technical definition.
An instance $(G_1,G_2)$ of the joint crossing number problem in $\Sigma$
is called {\em orientation-preserving homeo-invariant} if the input graphs
$G_1,G_2$ are given together with embeddings $G_1',G_2'$ in $\Sigma$, 
and the following holds:
there exists a joint embedding $G^0$ of $(G_1,G_2)$, achieving the joint crossing number, such that the subembedding of $G^0$
restricted to $G_i$ is orientation-preserving homeomorphic to $G_i'$,
for~$i=1,2$.

Note the important difference---while in the joint
orientation-preserving homeomorphic crossing number problem we require the
considered joint embeddings to respect the given homeomorphism classes of
$G^0_1,G^0_2$ (a {\em restriction}), 
for an orientation-preserving homeo-invariant instance we admit all
joint embeddings, but we know that some of the optimal solutions will respect the
homeo\-morphism classes of $G_1',G_2'$ (a {\em promise}).
We call {\sc OP-Homeo-Invariant Joint Crossing Number} problem
the ordinary {\sc Joint Crossing Number} problem with inputs restricted
only to orientation-preserving homeo-invariant instances.

%Clearly, if $\Sigma=\ca S_0$ is the sphere then always $\crgj{\ca S_0}(G_1,G_2)=0$.

\medskip

The following is a useful artifice 
in crossing numbers research. In a {\em weighted} graph, each edge is
assigned a positive number (the {\em weight, or thickness} of the edge). Now the 
{\em weighted joint crossing number} is defined as the ordinary
joint crossing number, but a crossing between edges $e_1$ and $e_2$,
say of weights $t_1$ and $t_2$, contributes $t_1 t_2$ to the weighted
joint crossing number. The weighted variants of the {joint
  homeomorphic crossing number} and of the 
{joint orientation-preserving homeomorphic crossing} numbers are defined
analogously.
%%%%11
\ifnoaccumulating{\par
In order to be able to smoothly use the weighted joint crossing number
variants in this paper, we give the following reduction,
which is easily proved using folklore tricks for
transforming weighted graphs into ordinary graphs.
% see for instance~\cite{DBLP:journals/siamcomp/CabelloM13}. (See also
% Proposition~\ref{pro:face-anchored-3} below for a related artifice, in
% which the restriction to $3$-connected graphs is worked out).
}\ifaccumulating{%
The following reduction is easily proved using folklore tricks for
transforming weighted graphs into ordinary ones; 
the fact that we are not in the plane does not play a role here.
% see for instance~\cite{DBLP:journals/siamcomp/CabelloM13}.
}

\ifaccumulating{\begin{proposition}[folklore]}
\ifnoaccumulating{\begin{proposition}}
\label{pro:weighted}
There is a polynomial-time reduction from the weighted joint crossing number
problem, with edge weights encoded in unary,
to the unweighted joint crossing number problem.
Moreover, this reduction can preserve $3$-connectivity and simplicity of the graphs.
\end{proposition}
\ifnoaccumulating{%
\begin{proof}
Consider an instance of the weighted {\sc Joint Crossing Number}
problem, that is, a pair of connected graphs $(G_1,G_2)$
and their edge-weight mappings $w_1$ and $w_2$, respectively.

We construct a graph $G_i'$, $i=1,2$, as follows:
every edge $e\in E(G_i)$ is replaced with a bunch $B_e$ of $w_i(e)$ parallel
unweighted edges in $G_i'$.
For every joint embedding $G^0$ of $(G_1,G_2)$ of weigthed crossing
number~$x$, we get a corresponding joint embedding of $(G_1',G_2')$
of crossing number equal to~$x$ by routing each bunch $B_e$ closely along $e$
as in~$G^0$.
Conversely, assume a joint embedding $G^1$ of $(G_1',G_2')$ of
crossing number~$x$.
For every $e\in E(G_1)\cup E(G_2)$, we choose an edge $e'\in B_e\subseteq
E(G_1')\cup E(G_2')$ which minimizes the number of crossings in $G^1$ among
all edges in~$B_e$.
The chosen edges $e'$ then define a joint embedding $G^0$ of $(G_1,G_2)$,
and the weighted number of crossings in $G^0$ is clearly at most~$x$.

Finally, in order to preserve simplicity and/or $3$-connectivity in the
reduction, we construct $G_i''$, $i=1,2$, from $G_i'$ as follows:
for every $e\in E(G_i)$, we subdivide every edge in $B_e\subseteq E(G_i')$
with a new vertex and connect these new vertices in $B_e$ together by a path
(in any order).
Since $G_i''$ contains a subdivision of $G_i$ and the added paths can
be easily drawn without crossings, this move does not change 
an optimum solution to the joint crossing number problem.
\qed\end{proof}
}

\begin{accumulate}
\ifaccumulating{\subsubsection*{An ordering lemma}}
Finally, this simple lemma will also be very useful:

\begin{lemma}\label{lem:ordering}
Let $a_1<a_2<\dots<a_k$ and $b_1>b_2>\dots>b_k$ be two sequences of integers.
Let $\pi$ be any permutation of $\{1,\dots,k\}$ other than the identity.
Then $\sum_{i=1}^ka_ib_{\pi(i)}-\sum_{i=1}^ka_ib_{i}\,\geq\,
	\min_{i\not=j}|a_i-a_j|$.
% Then the sum $\sum_{i=1}^ka_ib_{\pi(i)}$ is minimized if, and only if,
% the permutation $\pi$ is the identity.
\end{lemma}
\begin{proof}
% Assume the sum is minimized for some permutation~$\pi$.
If $\pi$ is not the identity, then there are indices $c<d$ such that $\pi(c)>\pi(d)$.
Let $\pi'$ be defined as follows;
$\pi'(c):=\pi(d)$, $\pi'(d):=\pi(c)$, and $\pi'(i):=\pi(i)$ otherwise.
Then
$$
\sum_{i=1}^ka_ib_{\pi(i)} - \sum_{i=1}^ka_ib_{\pi'(i)}
 = a_cb_{\pi(c)}+a_db_{\pi(d)} - a_cb_{\pi'(c)}-a_db_{\pi'(d)} =
$$ $$
 = a_cb_{\pi(c)}+a_db_{\pi(d)} - a_cb_{\pi(d)}-a_db_{\pi(c)}
 = (a_d-a_c)(b_{\pi(d)}-b_{\pi(c)}) \geq a_d-a_c>0
.$$
% a contradiction to the assumed minimal choice of~$\pi$.
The rest follows by induction (on the number of inversions in $\pi$).
\qed\end{proof}
\end{accumulate}

\section{Face-anchored Joint Embeddings}\label{sec:face-anchored}
%%%%%%%%%%%%%%%%%%%%%%%%%%%%%%%%%%%%%%%%%%%%%%

For the purpose of intermediate reduction we introduce the following variant
of the concept of joint embedding of $(G_1,G_2)$.
Assume that $C_1,\dots,C_k$ are cycles of the graph $G_1$ such that there
exists an embedding of $G_1$ in $\Sigma$ in which each of $C_1,\dots,C_k$ is
a facial cycle.
Let $a_1,\dots,a_k\in V(G_2)$.
A joint embedding $G^0$ of $(G_1,G_2)$ in $\Sigma$ is called
{\em face-anchored with respect to $\{(C_i,a_i): i=1,\dots,k\}$},
if the restriction of $G^0$ to $G_1$ contains a face $\alpha_i$ bounded by
$C_i$ such that the vertex $a_i$ of $G_2$ is drawn inside $\alpha_i$,
for all $i=1,\dots,k$.
The pairs $(C_i,a_i)$ are the {\em face anchors} of this joint embedding
\ifnoaccumulating{(or of the corresponding joint crossing number) }%
problem, where each $\alpha_i$ bounded by $C_i$ is an {\em anchor face} 
and each $a_i$ is an {\em anchor vertex}.

We will consider face-anchored joint embeddings
and their crossing number only in the case of $\Sigma$ being the
sphere $\ca S_0$ and $k$ being a constant, 
and then we specifically speak about {\em face-anchored joint planar
embeddings}, and call the corresponding algorithmic problem
{\sc $k$-FA Joint Planar Crossing Number}.
If inputs of this problem are restricted only to instances which are
orientation-preserving homeo-invariant (cf.~Section~\ref{sec:basic}),
then we speak about the
{\sc OP-Homeo-Invariant $k$-FA Joint Planar Crossing Number} problem.

% We will use it as an intermediate problem in our chain of hardness reductions.
% The following establishes the ``first half'' of our main result.

\begin{theorem}\label{thm:face-anchored}
For every integer $h\geq1$,
there is a polynomial-time reduction from the
{\sc OP-Homeo-Invariant $h$-FA Joint Planar Crossing Number} problem
to the {\sc OP-Homeo-Invariant Joint Crossing Number} problem
in the surface $\ca S_h$.
This reduction preserves connectivity of the involved graphs%
\ifaccumulating{, and the target problem can then 
be restricted to simple $3$-connected graphs}.
% For every integer $h\geq1$;
% there is a polynomial-time reduction from the face-anchored joint planar
% crossing number problem with $h$ face anchors, to the
% joint crossing number problem in the surface $\ca S_h$.
% Furthermore, the statement remains true under any combination of the
% following:
% \begin{enumerate}[a)]
% \item the destination joint crossing number problem may be restricted to
% its homeomorphic or orientation-preserving homeomorphic variants,
% \item\label{it:simpleconn}
% the destination problem may be restricted to connected and/or simple
% graphs if the input graphs of the source problem are such.
% \end{enumerate}
\end{theorem}

\begin{proof}
By Proposition~\ref{pro:weighted}, we may consider 
the source crossing problem as {\em unweighted} and to reduce to the {\em weighted}
joint crossing number problem in~$\ca S_h$, as long as the weights are
polynomial in the input size.

Consider an unweighted input $(G_1,G_2)$ of the 
{\sc OP-Homeo-Invariant $h$-FA Joint Planar Crossing Number} problem,
given along with the $h$ face anchors $\{(C_i,a_i): i=1,\dots,h\}$,
and with planar embeddings $G_1',G_2'$ of $G_1,G_2$
witnessing the homeo-invariant property.
To prove the theorem it suffices to construct (in polynomial time) 
a pair $(H_1',H_2')$ of $\ca S_h$-embedded graphs such that,
denoting by $H_1,H_2$ the corresponding abstract graphs, the following
holds:
\begin{itemize}
\item if $s$ is the (unknown) face-anchored joint planar crossing number
of $(G_1,G_2)$, then the joint orientation-preserving homeomorphic
(weighted) crossing number of $(H_1',H_2')$ is at most $f(s)$
(for a suitable function~$f$); and
\item if the joint crossing number of $(H_1,H_2)$ is at most
$f(s)$ for some integer~$s$, then the face-anchored joint planar crossing
number of $(G_1,G_2)$ is at most~$s$.
\end{itemize}

\definecolor{lightgray}{gray}{0.83}
\definecolor{llightgray}{gray}{0.93}
\begin{figure}[t]
\begin{center}%\vspace*{-1ex}
\begin{tikzpicture}[scale=0.4]
\tikzstyle{every path}=[line width=1.2pt, color=gray]
\draw[fill=llightgray,draw=none] (-4,-3) -- (13,-3) -- (12.5,6.6) -- (-1.5,6.6);
\draw[fill=white,draw=none] (5,3.4) ellipse (5.5 and 4.1);
\draw[fill=lightgray,draw=none] (5,0.1) ellipse (3.5 and 0.65);
% \draw (1.5,0) to [out=150,in=199] (2.6,6.8) to [out=19,in=161] (7.4,6.8) to [out=341,in=30] (8.5,0);
\draw (1.5,0) arc (230:-50:5.4 and 4.45);
\draw[dashed] (8.5,0.1) arc (0:180:3.5 and 0.6);
\draw (1.5,0.1) arc (180:360:3.5 and 0.7);
\draw[fill=llightgray] (5,3.8) ellipse (2.2 and 1.5);
\tikzstyle{every path}=[thin, color=blue]
\draw[dashed] (5,3.8) ellipse (2.6 and 1.8);
\draw (5,3.8) ellipse (3.0 and 2.1);
\draw (5,3.8) ellipse (3.6 and 2.5);
\draw (5,3.8) ellipse (4.1 and 2.9);
% \draw (5,3.8) ellipse (4.1 and 2.9);
\draw (1.45,1.7) arc (220:-40:4.6 and 3.3);
% \draw[rotate=-30] (6.1,6.3) ellipse (1.65 and 0.6);
\draw[dashed,rotate=-30] (7.8,6.3) arc (0:180:1.65 and 0.6);
\draw[rotate=-30] (7.8,6.3) arc (0:-180:1.65 and 0.6);
% \draw[rotate=30] (2.5,1.3) ellipse (1.65 and 0.6);
\draw[dashed,rotate=30] (4.15,1.3) arc (0:180:1.65 and 0.6);
\draw[rotate=30] (4.15,1.35) arc (0:-180:1.65 and 0.6);
\draw[dashed] (2.8,3.8) arc (0:180:1.6 and 0.4);
\draw (2.8,3.8) arc (0:-180:1.6 and 0.4);
\draw[dashed] (10.37,3.8) arc (0:180:1.6 and 0.4);
\draw (10.37,3.8) arc (0:-180:1.6 and 0.4);
\draw[dashed] (5,7.85) arc (90:270:0.2 and 1.3);
\draw (5,7.85) arc (90:-90:0.2 and 1.3);
% \draw[rotate=30] (9.45,1) ellipse (1.35 and 0.3);
\draw[dashed,rotate=30] (11,1) arc (0:180:1.45 and 0.3);
\draw[rotate=30] (11,1) arc (0:-180:1.45 and 0.3);
% \draw[rotate=-30] (-0.95,5.6) ellipse (1.35 and 0.3);
\draw[dashed,rotate=-30] (0.4,5.6) arc (0:180:1.45 and 0.3);
\draw[rotate=-30] (0.4,5.6) arc (0:-180:1.45 and 0.3);
\draw[dashed] (7.35,0.6) arc (-61:242:4.9 and 3.55);
\tikzstyle{every node}=[draw, shape=circle, inner sep=0.9pt, fill=red]
\tikzstyle{every path}=[line width=1.7pt, color=red]
\draw (7.1,-0.6) node[inner sep=1.9pt,label=below:$~~a_i\!\!$] (k1) {}
	 -- (7.2,0.4) node (k2) {} -- (6.6,0.9) node (k3) {}
	 -- (6.4,0.15) node[label=below:{\mbox{\large\boldmath$L_i\quad$}}] (k4) {}
	 -- (k1) ;
\draw (k3) to [out=150,in=210] (6,2.45) node (k5) {};
\draw[dashed] (k5) to [out=0,in=110] (8.1,0.4) node[draw=none] (k6) {}
		-- (8.9,0.44) {};
\draw (9.1,0.44) to [out=340,in=350] (k1);
\draw (3.5,0.2) node (k7) {} to [out=350,in=190] (k4);
\tikzstyle{every path}=[thin, color=red]
% \draw (5,3.6) ellipse (5 and 3.7);
\draw (7.1,0.38) arc (-65:285:5.1 and 3.7);
\draw (k5) to [out=200,in=90] (k7);
\tikzstyle{every path}=[thick, color=blue]
\tikzstyle{every node}=[draw, shape=circle, inner sep=1.1pt, fill=blue]
\draw (-4,-3) node[label=above:$\qquad C_i\qquad$] (b1) {} 
	 -- (13,-3) node (b2) {} -- (12.5,6.6) node (b3) {};
\draw (b1) node[label=below:\mbox{\large\boldmath$G_1\quad$}] {};
\draw (-1.5,6.6) node (b4) {} -- (b1);
\draw (b4) -- (0.9,6.6);
\draw (b3) -- (9,6.6);
\draw (b1) -- (-3.5,-5);
\draw (b2) -- (13,-5);
\draw (b2) -- (19,-2) node (b5) {} -- (20,6.5) -- (b3);
\draw (b5) -- (21,-3);
\tikzstyle{every path}=[thin, color=blue]
\tikzstyle{every node}=[draw, shape=circle, inner sep=0.7pt, fill=blue]
% \draw (b3) -- (9.55,5.85);
\draw (5,3.8) node[fill=white,draw=none] {\mbox{\large\boldmath$T_i$}};
\draw[dashed] (b4) -- (2.6,0.67); \draw (b4) -- (-0.3,4.85);
\draw (b2) -- (8.4,0) to [out=100,in=220] (8.5,1.7);
\draw (b1) -- (1.6,-0.1) to [out=70,in=-45] (1.4,1.77);
\draw (-3.3,-2.6) node[label=above:$\qquad\quad C_i'$] (bb1) {}
	 -- (12.4,-2.6) node {}	-- (12,6.2) node (bb3) {} -- (9.4,6.2);
\draw (bb1) -- (-1.15,6.1) node {} -- (0.5,6.1);
\draw (bb3) -- (b3);
\draw[dashed] (7.35,0.6) -- (bb3); \draw (10.5,4.4) -- (bb3);
\tikzstyle{every path}=[thick, color=red]
\tikzstyle{every node}=[draw, shape=circle, inner sep=1.1pt, fill=red]
\draw[->] (k1) -- (7.2,-5);
\draw (k1) to [out=320,in=190] (16,2) node (r1) {};
\draw[->] (r1) -- (17,-5);
\draw[->] (r1) -- (22,1);
\draw (18,-4.6) node[fill=none,draw=none] {\mbox{\large\boldmath$G_2$}};
\end{tikzpicture}
% \vspace*{-3ex}
\end{center}
\caption{A schematic detail of replacing one face anchor with a toroidal
gadget, as used in the proof of Theorem~\ref{thm:face-anchored}
(the torus attaches to the light-gray face via the gray hole).}
\label{fig:Tfaceanchors}
\end{figure}
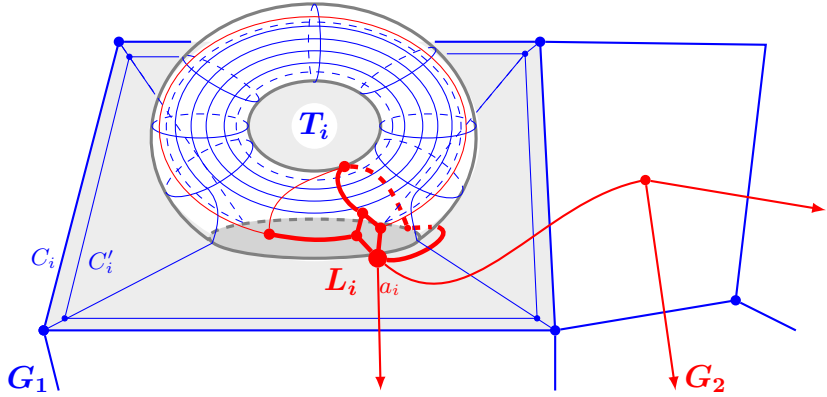

We may assume that each cycle $C_i$ is of length at least~$4$
(otherwise, we just subdivide it).
Our construction of $(H_1',H_2')$ can be shortly outlined as follows.
\begin{enumerate}[i)]%\parskip 2pt
\item
We assign to every edge of $G_1+G_2$ the same suitable weight $p$ (``medium thick'').
The purpose is that already a change in one crossing between $G_1$ and $G_2$
would cause a difference of $p^2$ in the target problem, a value larger than
all future required crossings between ``light'' edges of weight~$1$
and other edges of weight up to~$p$.
\item
% Consider the face anchors $(C_1,a_1),\dots,(C_h,a_h)$
% of a face-anchored joint planar embedding instance $(G_1,G_2)$.
For each $i=1,\dots,h$, we create a disjoint copy $C_i'$ of
weight~$1$ of the anchor cycle $C_i$,
and connect each vertex of $C_i'$ with its master copy in $C_i$.
Informally, we ``frame'' the $G_1'$-face bounded by $C_i$ with
$C_i'$ to force a unique plane subembedding, as in Figure~\ref{fig:Tfaceanchors}.
Let $G_1^+$ denote the resulting graph. % plane embedded graph.

Then we create a graph $T_i$ as follows.
Let $T_i^0$ be a new embedded graph made of a suitable toroidal grid 
after deleting specific two nonadjacent edges incident with the same $4$-cycle, 
to form an $8$-face in it.
$T_i$ is made of the existing cycle $C_i'$ and new $T_i^0$
by connecting the four degree-$3$ vertices of $T_i^0$ with
some four vertices of $C_i'$ in a matching cyclic order
(see again Figure~\ref{fig:Tfaceanchors}, the blue graph).
All the edges of $T_i$ have weight~$1$.
Let $H_1$ denote the resulting graph---the union of $G_1^+$ and of all
$T_i$, for $i=1,\dots,h$.
Note that $H_1$ has an embedding $H_1'$ in the surface $\Sigma\simeq\ca S_h$ 
obtained by adding one toroidal handle to each face bounded by~$C_i'$.
% (in the subembedding of $G_1$ in $G^0$).
% add a handle to the face $\alpha_i$ bounded by~$C_i$ and embed
% a new graph $T_i$ made of a suitable toroidal grid into this handle.
\item\label{it:refcx}
For each $i=1,\dots,h$, we create a new graph $L_i$ which is a copy of
$K_{3,3}$ with seven of its edges (except two incident ones)
made ``very thick'' of weight~$t_i$.
Let $H_2$ denote the graph made of $G_2$ and all $L_i$ after identifying one
vertex of $L_i$ with $a_i$, for $i=1,\dots,h$.
Then $H_2$ has an embedding $H_2'$ in~$\Sigma$,
such that $G_2'$ is a subembedding of~$H_2'$.
See the red graph in Figure~\ref{fig:Tfaceanchors}.

The informal purpose of such construction is two-fold;
first, the nonplanar graph $L_i$ must ``use'' some of the handles of $\Sigma$,
% by Proposition~\ref{pro:k33handle}, %%% Appendix!
and second, the thick edges of $L_i$ cannot cross any edge of
$G_1'$ which now have weight~$p$.
Consequently, each $L_i$ is ``confined'' to one of the $G_1'$-faces 
$\alpha_j$ bounded by $C_j'$.
Moreover, it will be shown that no two $L_i,L_{i'}$ for $i\not=i'$ 
are confined to the same face~$\alpha_j$.
\item
Additional detailed arguments ensure 
that \ref{it:refcx}) actually confines $L_i$, and hence also the anchor
vertex $a_i$, to $\alpha_i$ for $i=1,\dots,h$.
Briefly explaining this argument: for a sufficiently large integer $t$ we choose
$t_i:=(h+1-i)\cdot t$, and we choose the grid in each $T_j$ gadget 
such that the least number of edges of $T_j$ that have to be crossed 
by a noncontractible loop on the toroidal handle of~$T_j$ equals~$g_j:=5+j$.
\ifaccumulating{%
It is then an easy exercise in calculus to argue that the joint crossing number
is minimized only if ``$t_i$ is matched with~$g_i$''.
}\ifnoaccumulating{%
Then we finish by Lemma~\ref{lem:ordering} since $t$ is very large.}
\end{enumerate}
In other words, informally, an optimal joint embedding solution of $(H_1,H_2)$ must
``contain'' a feasible solution of $(G_1,G_2)$, and an optimal
orientation-preserving homeomorphic solution of $(G_1',G_2')$ ``generates'' 
a good orientation-preserving homeomorphic solution of $(H_1',H_2')$.
\ifaccumulating{Further technical details clarifying the stated proof outline,
and making the target graphs $3$-connected,
are left for the full preprint~\cite{arxivJoint} due to space restrictions.
}
\begin{accumulate}
\ifaccumulating{\subsection{Proof of Theorem~\ref{thm:face-anchored};
	technical details}
	\begin{proof}[continued]
}\ifnoaccumulating{\par\medskip}%
It remains to provide the details and prove correctness of the construction.
For $i=1,\dots,h$, the gadget $T_i$ results from the cycle $C_i'$
(of $G_1^+$, defined above) and a new graph $T_i^0$ made of
a toroidal grid of size $g_j\times(h+6)$ after removing some two edges $e,e'$ where
$e,e'$ are from two $(h+6)$-cycles and belong to the same quadrangle of the grid.
Note that $g_j<h+6$.
Then the four endvertices of former $e,e'$ are joined by
four edges to arbitrary four vertices of $C_i'$ (which is of length $\geq4$, see above)
in a matching cyclic order.
$T_i$ is naturally embedded in the torus as in Figure~\ref{fig:Tfaceanchors},
for $i=1,\dots,h$, and
together the given plane embedding $G_1'$ of $G_1$ this uniquely determines
the embedding $H_1'$ of $H_1$ in~$\ca S_h$.

\vbox{\smallskip
The two key properties of $T_i$ are as follows:
\vspace*{-1ex}
\begin{enumerate}[(T1)]
\item\label{it:Tnonpl} $T_i-V(C_i')$ is nonplanar;
\item\label{it:Tdew}
$T_i$ contains a subdivision of a $g_j\times(h+6)$ toroidal grid, and
hence in any embedding of $T_i$ in the torus, the least number of edges
crossed by a noncontractible loop is at least $g_j$.
The embedding induced by $H_1'$ and depicted in Figure~\ref{fig:Tfaceanchors} 
achieves this lower bound.
\end{enumerate}
}

Recall that $H_2$ results by a disjoint union of $G_2$ and the gadgets $L_i$
after identifying one vertex of $L_i$ with the anchor vertex $a_i$, for $i=1,\dots,h$.
The given embedding $G_2'$ of $G_2$ together with the embedding
of each $L_i$ as depicted in Figure~\ref{fig:Tfaceanchors} then determines
the embedding $H_2'$ of~$H_2$ in~$\ca S_h$.

Let $m=|E(G_1)|\cdot|E(G_2)|$.
We choose the weights in our construction as $p:=8m$ and $t:=(m+1)p^2$.
Let $s$ be the face-anchored joint crossing number of $(G_1,G_2)$,
and let $G'$ be a face-anchored joint embedding of $(G_1,G_2)$
with $s$ crossings (optimum) such that $G'$ restricted to $G_j$ is
orientation-preserving homeomorphic to given~$G_j'$, for $j=1,2$.
Note that $s\leq m$.
Any such $G'$ can be easily extended to a joint orientation-preserving
homeomorphic embedding $H'$ of $(H_1',H_2')$.
We start with an estimate of the weighted crossing number of~$H'$.

% Recall that, formally, $G_1\subseteq H_1$ and $G_2\subseteq H_2$, but the
% edges of $G_i$ all have weight $p$ in $H_i$, $i=1,2$.
% We first compute the weighted crossings in the following joint embedding
% $J^0$ of $(H_1,H_2)$ in $\ca S_h$:
% the subgraphs $G_1\subseteq H_1$ and $G_2\subseteq H_2$ are drawn as in an
% optimal solution of the face-anchored instance $(G_1,G_2)$,
% and the gadgets of each anchor face are drawn in the natural way as depicted
% in Figure~\ref{fig:Tfaceanchors}.

Recall that the edges of $G_j$ all have weight $p$ in $H_j$, $j=1,2$.
The subdrawing of $G_1+G_2$ in $H'$ thus contributes precisely
$s\cdot p^2$ to the total crossing number.
For $i=1,\dots,h$, the subdrawing of $T_i+L_i$ contributes at most
$g_i\cdot t_i+h+6+g_i$ and that of $T_i+G_2$ contributes at most
$2p\cdot d_{G_2}(a_i)$.
The weighted crossing number of $H'$ hence can be
estimated from above (with a large margin) by
\begin{eqnarray}
s\cdot p^2 &+& \sum_{i=1}^h g_i\cdot t_i +2h(h+6) +
	 2p\cdot \sum_{i=1}^h d_{G_2}(a_i) \leq
\nonumber\\\label{eq:Jupper}
&\leq& s\cdot p^2 + \sum_{i=1}^h g_i\cdot t_i +2m^2 + 2pm \leq
	s\cdot p^2 + \sum_{i=1}^h g_i\cdot t_i + p^2/2
.\end{eqnarray}

Conversely, we would like to estimate $s$ in terms of the joint 
crossing number $r$ of $(H_1,H_2)$ in the surface $\Sigma\simeq\ca S_h$.
\medskip

Consider a joint embedding $H^0$ of $(H_1,H_2)$ in $\Sigma$
of weighted crossing number~$r$.
Let $H_1^0$ be the $\Sigma$-embedding of $H_1$ induced by $H^0$.
Then, by (T\ref{it:Tnonpl}) and Proposition~\ref{pro:disj-nonplanar-b},
the submebedding of $G_1^+$ in $H_1^0$ is plane and so the anchor faces $\alpha_i$
now bounded by $C_i'$ are well defined in $G_1^+$.
We cut the surface $\Sigma$ simultaneously along $C_i'$, $i=1,\dots,h$.
Let $\Sigma_i$ denote the resulting surface with a boundary $C_i'$ not
containing $C_i$.
Note that it might theoretically happen that $\Sigma_i=\Sigma_j$ for some
$i\not=j$, which means that there is one such subsurface incident with both
of $C_i'$ and $C_j'$ (informally, a handle might ``stretch'' 
from $\alpha_i$ to $\alpha_j$).

We first show that the latter case $\Sigma_i=\Sigma_j$, $i\not=j$, cannot happen.
The collection of surfaces $\{\Sigma_1,\dots,\Sigma_h\}$ (without
repetition) together embeds $h$ nonplanar pairwise-disjoint subgraphs
by (T\ref{it:Tnonpl}), and so the sum of their genera is at least $h$
by Proposition~\ref{pro:disj-nonplanar-a}.
If $\Sigma_i=\Sigma_j$ then the plane of $G_1^+$, when added back to
$\Sigma_i$, would act as an additional handle, making the genus of $\Sigma$
higher than~$h$, a contradiction.
So, $\Sigma_1,\dots,\Sigma_h$ are pairwise distinct and each of genus
exactly~$1$.
Let $\Sigma_i^+$ denote the union of $\Sigma_i$ and the open faces and edges
incident to $C_i'$, that is, the boundary not belonging to $\Sigma_i^+$ is
exactly the cycle~$C_i$.

Second, since $t\cdot p$ largely exceeds the estimate \eqref{eq:Jupper},
no thick edge of any $L_j$ gadget may cross any $C_i$ in the
supposedly optimal drawing~$H^0$.
Hence each nonplanar $L_j$ (except, possibly, of the two thin edges), 
$j=1,\dots,h$, is drawn in one of $\Sigma_1^+,\dots,\Sigma_h^+$.
By Proposition~\ref{pro:k33handle} and the fact that $\Sigma_i^+$ is of
genus~$1$ we get that no two distinct $L_j,L_{j'}$ are drawn in the
same~$\Sigma_i^+$.
Consequently, there is a permutation $\pi$ of $\{1,\dots,h\}$ such that
$V(L_{\pi(i)})$ belongs to $\Sigma_{i}^+$.

By Proposition~\ref{pro:k33handle}, at least one of the cycles of
$L_{\pi(i)}$ of weight $t_{\pi(i)}$ is drawn noncontractible in $\Sigma_{i}^+$,
and so it contributes at least $g_{i}\cdot t_{\pi(i)}$ by (T\ref{it:Tdew})
to the weighted crossing number of $H^0$.
Using the upper estimate \eqref{eq:Jupper} and the fact that
$t=(m+1)p^2>sp^2+p^2/2$, we see that $\sum_{i=1}^h g_i\cdot t_{\pi(i)}$
cannot exceed $\sum_{i=1}^h g_i\cdot t_i$ by $t$ or more,
which in turn by Lemma~\ref{lem:ordering} means that $\pi$ is the identity.
This concludes that $H^0$ restricted to $G_1+G_2$ is a feasible solution of
the face-anchored joint embedding instance $(G_1,G_2)$
with $s\leq\big(r-\sum_{i=1}^h g_i\cdot t_i\big)/p^2$ unweighted crossings.
Plugging this into the estimate \eqref{eq:Jupper} we get
$s=\big\lfloor\big(r-\sum_{i=1}^h g_i\cdot t_i\big)/p^2\big\rfloor$
which establishes correctness of our reduction.
\ifaccumulating{\qed\end{proof}}
\end{accumulate}
\qed\end{proof}

\begin{accumulate}
In order to raise the connectivity premise in
Theorem~\ref{thm:face-anchored}, we shall use the following additional reduction.
% \ifaccumulating{ (proved in the Appendix)}:
Note that this claim is stronger than former Proposition~\ref{pro:weighted}
since we (intentionally) do not assume $3$-connectivity of the original
instance, and that makes our task significantly harder than the former one.

\begin{proposition}\label{pro:face-anchored-3}
On every surface $\ca S_h$, $h\geq1$;
there is a polynomial-time reduction from the problem
{\sc OP-Homeo-Invariant Joint Crossing Number} of a pair of connected graphs,
to {\sc OP-Homeo-Invariant Joint Crossing Number} restricted
to pairs of simple $3$-connected graphs.
\end{proposition}

% \ifaccumulating{\subsection{Proof of Proposition~\ref{pro:face-anchored-3}}}
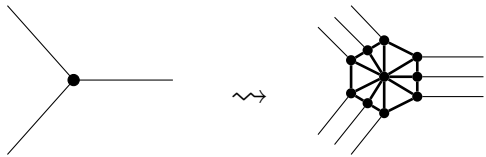
\begin{figure}[ht]
\begin{center}
\begin{tikzpicture}[scale=0.22]
\tikzstyle{every node}=[draw, shape=circle, inner sep=1.5pt, fill=black]
\draw (-2,-1.5) -- (2,3) node {} -- (-2,7.5);
\draw (2,3) -- (8,3);
\end{tikzpicture}
\qquad\raise2em\hbox{\Large$\leadsto$}\qquad
\begin{tikzpicture}[scale=0.22]
\tikzstyle{every node}=[draw, shape=circle, inner sep=1pt, fill=black]
\tikzstyle{every path}=[line width=1.0pt]
\draw (2,3) node (s) {};
\draw (0,2) node (a1) {}; \draw (1,1.4) node (a2) {}; \draw (2,0.8) node (a3) {};
\draw (0,4) node (b1) {}; \draw (1,4.6) node (b2) {}; \draw (2,5.2) node (b3) {};
\draw (4,4.2) node (c1) {}; \draw (4,3) node (c2) {}; \draw (4,1.8) node (c3) {};
\draw (a1) -- (a2) -- (a3) -- (c3) -- (c2) -- (c1) -- (b3) -- (b2) -- (b1) -- (a1);
\draw (a1) -- (s) -- (a2); \draw (a3) -- (s) -- (b1);
\draw (b2) -- (s) -- (b3);
\draw (c2) -- (s) -- (c3); \draw (s) -- (c1);
\tikzstyle{every path}=[thin]
\draw (a1) -- (-2,-0.5); \draw (a2) -- (-1,-1.1); \draw (a3) -- (0,-1.7);
\draw (b1) -- (-2,6); \draw (b2) -- (-1,6.6); \draw (b3) -- (0,7.3);
\draw (c1) -- (8,4.2); \draw (c2) -- (8,3); \draw (c3) -- (8,1.8);
\end{tikzpicture}
\end{center}
\caption{A local detail of the construction in the proof of
Proposition~\ref{pro:face-anchored-3}. The thick edges get assigned a
weight $10\cdot|E(G_1)||E(G_2)|$, and the thin edges are of weight of $1$.}
\label{fig:makeconn-3}
\end{figure}

\begin{proof}
Consider an instance of the {\sc OP-Homeo-Invariant Joint Crossing Number}
problem in $\ca S_h$, that is, a pair of connected graphs $(G_1,G_2)$
and their given embeddings $G_1',G_2'$ in $\ca S_h$.

% Let $G_1,G_2,{\ca S_h},k$ be an input of {\sc Joint Crossing Number};
% that is, $G_1$ and $G_2$ are connected graphs that embed in 
% ${\ca S_h}$, and the question is whether or not there is a joint embedding
% of $G_1+G_2$ with at most $k$ crossings.

We construct graphs $G_1^+,G_2^+$ as follows.
Let $d(v)$ denote the degree of a vertex $v\in V(G_1+G_2)$.
We start by ``blowing up'' every vertex $v\in V(G_1+G_2)$ into a wheel
$W(v)$ of size $3d(v)$ (that is, the hub of the wheel has degree $3d(v)$). 
Then we assign to the edges of each such wheel a weight
$10\cdot|E(G_1)|\cdot|E(G_2)|$, (cf.~Proposition~\ref{pro:weighted}). 
Then we replace every edge $e=uv\in E(G_1+G_2)$ by three edges of weight $1$ 
which join consecutive triples of rim vertices of $W(u)$ and $W(v)$ 
in the natural order, respecting the vertex rotation in the corresponding
embedding $G_1'$ or $G_2'$.
This transformation is illustrated in Figure~\ref{fig:makeconn-3}.

Let $(G_1^+,G_2^+)$ be the resulting pair of graphs.
It is straightforward to see that both $G_1^+$ and $G_2^+$ are simple and
$3$-connected, and they embed in $\ca S_h$.

Assume that the joint crossing number of $(G_1,G_2)$ in $\ca S_h$ is~$k$.
Then there exists a joint orientation-preserving homeomorphic embedding $G'$
of $(G_1',G_2')$ with $k$ crossings.
In the drawing $G'$, we choose a sufficiently small open neighbourhood of each vertex
$v\in V(G_1+G_2)$ and draw the wheel $W(v)$ in this neighbourhood, not
crossing any of the wheel edges.
For every edge $e=uv\in E(G_1+G_2)$ we draw the replacement three edges of
$G_1^++G_2^+$ in a small neighbourhood of the drawing of $e$.
This results in a joint embedding of $(G_1,G_2)$ with $3\cdot3\cdot k=9k$
crossings.

Conversely, consider an optimal joint embedding $G^0$ of $(G_1^+,G_2^+)$
in $\ca S_h$ with $\ell$ crossings.
Since obviously, $\ell\leq 9\cdot|E(G_1)|\cdot|E(G_2)|$, we know that none
of the edges of $W(v)$, $v\in V(G_1+G_2)$, is crossed.
By a standard argument, we may assume that all the three edges replacing
in $G_1^++G_2^+$ one edge $e$ of $G_1+G_2$, are routed along the same way.
Hence if we contract every wheel $W(v)$, $v\in V(G_1+G_2)$ into the vertex
$v$ and simplify the triples of resulting parallel edges, 
we obtain a joint embedding of $(G_1,G_2)$ in $\ca S_h$ with at most $\ell/9$ crossings.
\qed\end{proof}
\end{accumulate}

\section{Multiplying Face Anchors}\label{sec:mult}
%%%%%%%%%%%%%%%%%%%%%%%%%%%%%%%%%%%%%%%%%%%%%%%%%%%%%%%%%%%
\label{sec:multiplying}

Recall that our ultimate goal is to find a reduction from a special variant of
the anchored crossing number
problem~\cite{DBLP:journals/siamcomp/CabelloM13},
described in Section~\ref{sec:anchoredCM}.
This can already be achieved with Theorem~\ref{thm:face-anchored}, but such
an approach would require an unbounded number of face anchors
(and hence unbounded genus in the Joint crossing number problem).
We thus present the following construction which ``multiplies'' the number
of available face anchors, albeit in a special position.

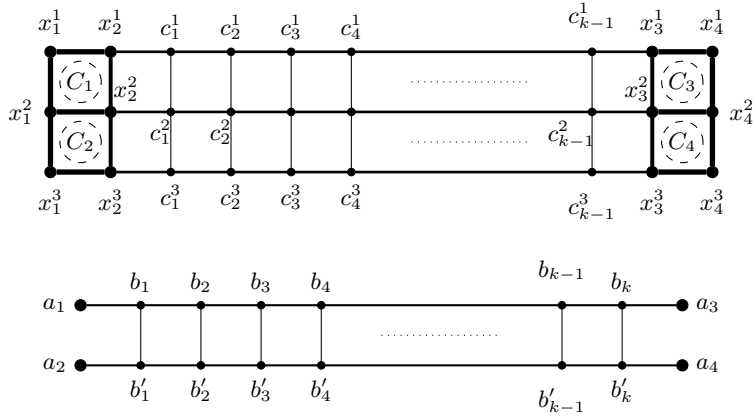
\begin{figure}[t]
\begin{center}%\vspace*{-3ex}
\begin{tikzpicture}[scale=0.4]
\tikzstyle{every node}=[draw, shape=circle, minimum size=2pt,inner sep=1.5pt, fill=black]
\draw (0,0) node[label=below:$x^3_1$] (x1) {}; \draw (2,0) node[label=below:$x^3_2$] (x2) {};
\draw (0,2) node[label=left:$x^2_1$] (y1) {}; \draw (2,2) node[label=above:$x^2_2\hspace*{-3ex}$] (y2) {};
\draw (0,4) node[label=above:$x^1_1$] (z1) {}; \draw (2,4) node[label=above:$x^1_2$] (z2) {};
\draw (20,0) node[label=below:$x^3_3$] (x3) {}; \draw (22,0) node[label=below:$x^3_4$] (x4) {};
\draw (20,2) node[label=above:$\hspace*{-3ex}x^2_3$] (y3) {}; \draw (22,2) node[label=right:$x^2_4$] (y4) {};
\draw (20,4) node[label=above:$x^1_3$] (z3) {}; \draw (22,4) node[label=above:$x^1_4$] (z4) {};
\tikzstyle{every path}=[line width=2pt]
\draw (x2) -- (x1) -- (z1) -- (z2);
\draw (x3) -- (x4) -- (z4) -- (z3);
\draw (y1) -- (y2); \draw (y3) -- (y4);
\tikzstyle{every path}=[line width=1.5pt]
\draw (x2) -- (z2); \draw (x3) -- (z3);
\tikzstyle{every path}=[line width=1.0pt]
\draw (x2) -- (x3); \draw (y2) -- (y3); \draw (z2) -- (z3);
\tikzstyle{every node}=[draw, shape=circle, minimum size=1pt,inner sep=1pt, fill=black]
\tikzstyle{every path}=[thin]
\draw (4,0) node[label=below:$c^3_1$] {} -- (4,2) node[label=below:$\hspace*{-2ex}c^2_1$] {} -- (4,4) node[label=above:$c^1_1$] {};
\draw (6,0) node[label=below:$c^3_2$] {} -- (6,2) node[label=below:$\hspace*{-2ex}c^2_2$] {} -- (6,4) node[label=above:$c^1_2$] {};
\draw (8,0) node[label=below:$c^3_3$] {} -- (8,2) node {} -- (8,4) node[label=above:$c^1_3$] {};
\draw (10,0) node[label=below:$c^3_4$] {} -- (10,2) node {} -- (10,4) node[label=above:$c^1_4$] {};
\draw (18,0) node[label=below:$c^3_{k-1}$] {} -- (18,2) node[label=below:$\hspace*{-4ex}c^2_{k-1}$] {} -- (18,4) node[label=above:$c^1_{k-1}$] {};
\tikzstyle{every path}=[dotted, thin]
\draw (12,1) -- (16,1); \draw (12,3) -- (16,3);
\draw (1,3) node[dashed,fill=none] {$C_1$};
\draw (1,1) node[dashed,fill=none] {$C_2$};
\draw (21,3) node[dashed,fill=none] {$C_3$};
\draw (21,1) node[dashed,fill=none] {$C_4$};
\end{tikzpicture}
\\
\begin{tikzpicture}[scale=0.4]
\tikzstyle{every node}=[draw, shape=circle, minimum size=2pt,inner sep=1.5pt, fill=black]
\draw (1,1) node[label=left:$a_2$] (a1) {};
\draw (1,3) node[label=left:$a_1$] (a2) {};
\draw (21,1) node[label=right:$a_4$] (a3) {};
\draw (21,3) node[label=right:$a_3$] (a4) {};
\tikzstyle{every path}=[thick]
\draw (a1) -- (a3); \draw (a2) -- (a4);
\tikzstyle{every node}=[draw, shape=circle, minimum size=1pt,inner sep=1pt, fill=black]
\tikzstyle{every path}=[thin]
\draw (3,1) node[label=below:$b'_1$] {} -- (3,3) node[label=above:$b_1$] {};
\draw (5,1) node[label=below:$b'_2$] {} -- (5,3) node[label=above:$b_2$] {};
\draw (7,1) node[label=below:$b'_3$] {} -- (7,3) node[label=above:$b_3$] {};
\draw (9,1) node[label=below:$b'_4$] {} -- (9,3) node[label=above:$b_4$] {};
\draw (17,1) node[label=below:$b'_{k-1}$] {} -- (17,3) node[label=above:$b_{k-1}$] {};
\draw (19,1) node[label=below:$b'_k$] {} -- (19,3) node[label=above:$b_k$] {};
\tikzstyle{every path}=[dotted]
\draw (11,2) -- (15,2);
\end{tikzpicture}
% \vspace*{-3ex}
\end{center}
\caption{The graphs $F_1$ (top) and $F_2$ (bottom) of the face-anchored joint planar embedding
	problem $\ca F_{k,T}$; the precise weights of the edges are
	specified in (F\ref{it:thickT3})--(F\ref{it:thlast}) below.}
\label{fig:F1F2}
\end{figure}

\begin{figure}[t]
\begin{center}\bigskip
\begin{tikzpicture}[scale=0.4]
\normalsize
\tikzstyle{every node}=[draw, shape=circle, minimum size=2pt,inner sep=1.5pt, fill=black]
\tikzstyle{every path}=[thick, color=blue]
\draw (0,0) node (x1) {}; \draw (2,0) node (x2) {};
\draw (0,2) node[label=left:\mbox{\boldmath$F_1$}] (y1) {}; \draw (2,2) node (y2) {};
\draw (0,4) node (z1) {}; \draw (2,4) node (z2) {};
\draw (20,0) node (x3) {}; \draw (22,0) node (x4) {};
\draw (20,2) node (y3) {}; \draw (22,2) node (y4) {};
\draw (20,4) node (z3) {}; \draw (22,4) node (z4) {};
\tikzstyle{every path}=[line width=2pt, color=blue]
\draw (x2) -- (x1) -- (z1) -- (z2);
\draw (x3) -- (x4); \draw (z4) -- (z3);
\draw (y1) -- (y2); \draw (y3) -- (y4);
\tikzstyle{every path}=[line width=1.5pt, color=blue]
\draw (x2) -- (z2); \draw (x3) -- (z3);
\draw (x4) -- (z4);
\tikzstyle{every path}=[line width=1.0pt, color=blue]
\draw (x2) -- (x3); \draw (y2) -- (y3); \draw (z2) -- (z3);
\tikzstyle{every node}=[draw, shape=circle, minimum size=1pt,inner sep=1pt, fill=black]
\tikzstyle{every path}=[thick, color=blue]
\draw (4,0) node {} -- (4,2) node {} -- (4,4) node {};
\draw (6,0) node {} -- (6,2) node {} -- (6,4) node {};
\draw (8,0) node {} -- (8,2) node {} -- (8,4) node {};
\draw (10,0) node {} -- (10,2) node {} -- (10,4) node {};
\draw (18,0) node {} -- (18,2) node {} -- (18,4) node {};
\tikzstyle{every path}=[thin, color=black]
\draw (1,3) node[dashed,fill=none, minimum size=3ex] {};
\draw (1,1) node[dashed,fill=none, minimum size=3ex] {};
\draw (21,3) node[dashed,fill=none, minimum size=3ex] {};
\draw (21,1) node[dashed,fill=none, minimum size=3ex] {};
\tikzstyle{every node}=[draw, shape=circle, minimum size=2pt,inner sep=1.5pt, fill=black]
\tikzstyle{every path}=[thick, color=red]
\draw (1,1) node (a1) {};
\draw (1,3) node (a2) {};
\draw (21,1) node[label=right:\mbox{\boldmath$\quad~F_2$}] (a3) {};
\draw (21,3) node (a4) {};
\draw (a1) -- (a3); \draw (a2) -- (a4);
\tikzstyle{every node}=[draw, shape=circle, minimum size=1pt,inner sep=1pt, fill=black]
\tikzstyle{every path}=[thin, color=red]
\draw (3,1) node {} -- (3,3) node {};
\draw (5,1) node {} -- (5,3) node {};
\draw (7,1) node {} -- (7,3) node {};
\draw (9,1) node {} -- (9,3) node {};
\draw (17,1) node {} -- (17,3) node {};
\draw (19,1) node {} -- (19,3) node {};
\end{tikzpicture}
\end{center}
\caption{Supposed crossing-optimal face-anchored joint planar embedding of
	$\ca F_{k,T}$.} % ($F_1$~in~blue and $F_2$ in red).}
\label{fig:F1F2joint}
\end{figure}
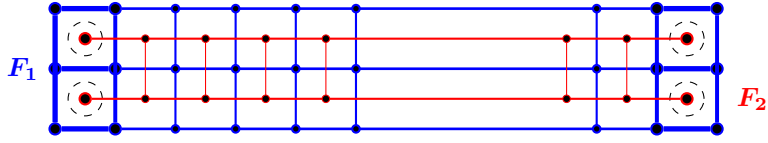

Let $F_1$ be the graph of the $3\times(k+3)$ plane grid,
and $F_2$ be obtained from the $2\times(k+2)$ plane grid by removing
the two side edges (making a ``ladder''), with notation as in Figure~\ref{fig:F1F2}.
Let $C_1$ denote the cycle $(x^1_1,x^1_2,x^2_2,x^2_1)$ of $F_1$,
and $C_2,C_3,C_4$ the cycles $(x^2_1,x^2_2,x^3_2,x^3_1)$,
$(x^1_3,x^1_4,x^2_4,x^2_3)$, $(x^2_3,x^2_4,x^3_4,x^3_3)$.
The weights of the edges of $F_1$ are as follows
(where $T$ is a large integer):
\begin{enumerate}[(F1)]
\item\label{it:thickT3}
weight $T^3$ for the six edges
$x^1_2x^2_2,x^2_2x^3_2,x^1_3x^2_3,x^2_3x^3_3,x^1_4x^2_4,x^2_4x^3_4$ and
% (where $T$ is large),
\item\label{it:thickT4}
weight $T^4$ for the remaining eight edges induced on the vertex set
$\{x^i_j:i\in\{1,2,3\},j\in\{1,2,3,4\}\}$
(yes, this part is intentionally not symmetric),
\item\label{it:thickT2}
weight $T^2$ for every ``horizontal'' edge on the shortest paths from $x^i_2$ to
$x^i_3$, for $i=1,2,3$,
\item\label{it:thickT1}
weight $jT$ for the ``vertical'' edges $c^1_{j}c^2_{j}$ and
$c^2_{k-j}c^3_{k-j}$, for \mbox{$j=1,2,\dots,k-1$}.
\end{enumerate}

The weights of the edges of $F_2$ are as follows:
\begin{enumerate}[(F1)]\setcounter{enumi}{4}
\item\label{it:horizk3}
weight $t_{j-1}$ for the ``horizontal'' edges $b'_{j-1}b'_{j}$ and $b_{k+2-j}b_{k+1-j}$,
for $j=1,2,\dots,k+1$ where $b_0=a_1,b'_0=a_2,b_{k+1}=a_3,b'_{k+1}=a_4$, and
$t_j$ is defined by $t_0=k^3$ and $t_j=t_{j-1}+j$,
\item\label{it:thlast} weight $k+1$ for all the ``vertical'' edges $b_{j}b'_{j}$,
for $j=1,2,\dots,k$.
\end{enumerate}

Finally, we shortly denote by $\ca F_{k,T}$ the joint planar embedding
instance of $(F_1,F_2)$ with the set of four face anchors $\{(C_i,a_i):i=1,2,3,4\}$.
\ifaccumulating{The details of the following claim 
 are left for the full preprint~\cite{arxivJoint} due to space restrictions:}
%%%%11

\begin{lemma}
\label{lem:F1F2}
For every sufficiently large $k$ and $T=\Omega(k^6)$,
every joint planar embedding solution of $\ca F_{k,T}$ other than
the one depicted in Figure~\ref{fig:F1F2joint}
has its weighted crossing number exceeding that of
Figure~\ref{fig:F1F2joint} by at least~$T$.
Moreover, if a solution of $\ca F_{k,T}$ draws any one of the vertices
$b_i$ or $b_i'$ for $i\in\{1,\dots,k\}$ in the $F_1$-face incident with
both $x_2^1,x_3^1$, then its weighted crossing number exceeds the optimum by
at least $\Omega(k^3)\cdot T^2$.
\end{lemma}

\begin{accumulate}
\ifaccumulating{\subsection{Proof of Lemma~\ref{lem:F1F2};
	technical details}}
% 	\begin{proof}[continued]}
\begin{proof}
Note that a planar embedding of $F_1$ itself is unique.
The weighted crossing number of the joint planar embedding solution to $\ca
F_{k,T}$ in Figure~\ref{fig:F1F2joint} is
{\small
\begin{eqnarray} \label{eq:T3beta}
T^3\cdot(2k^3&+&2(k^3+1+\dots+k)) + \beta(k,T) =
  (4k^3+k^2+k)T^3 + \beta(k,T)
,\\\nonumber\mbox{where }
\beta(k,T) &=\,& k(k+1)T^2 + 2\sum_{j=1}^{k-1}
		(k-j)T(k^3+1+\dots+j)
\\	&=&  k(k+1)T^2 + T\sum_{j=1}^{k-1}(2k^4+kj^2+kj-2k^3j-j^3-j^2)
\nonumber
\\\label{eq:T2gamma}	&=&
\ifnoaccumulating{  (k^2+k)T^2 + \gamma(k)T =}
		 \ca O(k^2)\,T^2 + \ca O(k^5)\,T
\ifnoaccumulating{%
\\\nonumber\mbox{since }
\gamma(k) &=\,& 2k^4(k-1) -\sum_{j=1}^{k-1}j^3 +(k-1)\sum_{j=1}^{k-1}j^2
		-(2k^3-k)\sum_{j=1}^{k-1}j
\\      &=&  2k^4(k-1) -\frac14k^2(k-1)^2 +\frac16k(k-1)^2(2k-1)
		-\frac12k(k-1)(2k^3-k)
\nonumber
\\      &=&  \frac1{12}\big(12k^5-11k^4+2k^3-k^2-2k\big) = \ca O(k^5)
\nonumber.
}%
\end{eqnarray}
}%

The weight of the smallest edge cut separating any one of 
$a_2,a_3$ from the remaining three anchor vertices is $t_0=k^3$,
and one can argue by case-checking that the smallest edge cut separating 
any one of  $a_1,a_4$ from the remaining three anchor vertices is of weight 
$t_{k}=k^3+\frac12(k^2+k)$.
Since we may easily assume that the $F_1$-edges listed in (F\ref{it:thickT4})
are not crossed due to their extreme weight $T^4>\!\!>T^3$,
we can directly conclude that the unavoidable crossings between the edges
of $F_2$ and the edges of $F_1$ listed in (F\ref{it:thickT3})
contribute at least $(t_0+t_0+t_k+t_k)\cdot T^3=(4k^3+k^2+k)\cdot T^3$ to
the total sum, as in Figure~\ref{fig:F1F2joint}.
% Hence, unavoidably, each one of the shortest paths $Q$ and $Q'$
% in $F_2$ joining $a_1$--$a_3$ and $a_2$--$a_4$, respectively,
% makes at least $\frac12(4k^3+k^2+k)T^3$ crossings with the ``very thick'' edges
% of $F_1$ specified in (F\ref{it:thickT3}),(F\ref{it:thickT4}) above.

Comparing the latter quantity to \eqref{eq:T3beta} and \eqref{eq:T2gamma}
we see that, for sufficiently large~$k$,
we can assume that there is no other crossing (than what is mentioned in the
previous paragraph) of $F_2$ with the edges listed in 
(F\ref{it:thickT3}),(F\ref{it:thickT4}).
In particular, the shortest (``horizontal'') path $Q_i$, $i=1,2$,
in $F_2$ connecting $a_i$ to $a_{i+2}$,
crosses the two edges $x_2^{i}x_2^{i+1}$,$x_3^{i}x_3^{i+1}$ but not, e.g.,
the edge $x_4^{i}x_4^{i+1}$.
Moreover, by a finer resoultion of \eqref{eq:T2gamma}, we see
that either there is no crossing between the edges of $F_1$ specified in (F\ref{it:thickT2})
and the edges of $F_2$ listed in (F\ref{it:horizk3}),
or that the optimum is exceeded by at least 
$(t_0-\ca O(k^2))\cdot T^2\geq\Omega(k^3)\cdot T^2$ weighted crossings, as desired.
Consequently, in what follows we may assume 
that the path $Q_i$ is contained in the region bounded by
$(x_1^ix_2^ic_1^i\dots c_{k-1}^ix_3^ix_4^ix_4^{i+1}x_3^{i+1}\dots
c_1^{i+1}x_2^{i+1}x_1^{i+1}x_1^i)$.

Informally, we have reached a situation with a joint embedding quite similar
to that of Figure~\ref{fig:F1F2joint}, only that we do not know how the
non-anchor vertices of $F_2$ are distributed into the square faces of~$F_1$
(while we aim for the natural ordered one-to-one assignment).
The final technical step in the proof can now be achieved using 
arguments quite similar to those of Lemma~\ref{lem:ordering}.
% \ifaccumulating{The details are again left for the Appendix due to space restrictions.}
%%%%11

\medskip
Let $c^i_0=x^i_2$ and $c^i_k=x^i_3$ for $i=1,2,3$.
Let $B_1,\dots,B_k$ and $B'_1,\dots,B'_k$ denote
(from left to right) the non-anchor square faces of $F_1$;
where $B_j$ is the face bounded by $(c^1_{j-1},c^1_j,c^2_j,c^2_{j-1})$
and $B'_j$ is bounded by $(c^2_{j-1},c^2_j,c^3_j,c^3_{j-1})$.
%  such that the intersection of $B_{i-1}$ and $B_i$ is the edge
% $c^1_ic^2_i$ and the intersection of $B'_{i-1}$ and $B'_i$ is $c^2_ic^3_i$.
Let $\iota(j)$ be such that the $F_2$-vertex $b_j$ is drawn inside the
$F_1$-face $B_{\iota(j)}$, and $\iota'(j)$ be such that $b'_j$ is drawn
inside $B_{\iota'(j)}$.
We aim to show that $\iota(j)=\iota'(j)=j$ for $j=1,\dots,k$,
or the crossing number of the considered joint embedding solution is by at
least $T$ more than the optimum~\eqref{eq:T3beta}.
This will follow by a straightforward induction if we prove that in any case
violating $\iota(j)=\iota'(j)=j$ there is a local change in the joint
embedding which decreases the crossing number by at least~$T$.

Let $R^i$, $i=1,2,3$, denote the shortest ``horizontal'' path in $F_1$
from $x^i_2$ to $x^i_3$, which have weights $T^2$ by (F\ref{it:thickT2}).
First, we argue that each of the ``vertical'' $F_2$-edges $b_jb'_j$ crosses
only one edge of $R^1\cup R^2\cup R^3$ (and so it crosses $R^2$).
Suppose not, then we can redraw $b_jb'_j$ across $R^2$ and its incident
edges, using at most $T^2+k^2T<2T^2-T$ weighted crossings---see
(F\ref{it:thickT2}),(F\ref{it:thickT1}).
This new drawing hence saves at least $T$ weighted crossings on $b_jb'_j$,
as needed for our inductive argument.
Second, if (up to symmetry) $\iota(j)\leq\iota'(j)$ and $b_jb'_j$ crosses
the ``vertical'' $F_1$-edge $c^2_{\iota'(j)}c^3_{\iota'(j)}$, then
we can again redraw $b_jb'_j$ with saving more than $T$ weighted crossings.

Now, assume that $\iota(j)\not=\iota'(j)$ for some $j\in\{1,\dots,k\}$.
By the previous, the ``vertical'' $F_2$-edge $b_jb'_j$ cannot cross both of
the $F_1$-edges $c^2_{\iota(j)-1}c^2_{\iota(j)}$,
$c^2_{\iota'(j)-1}c^2_{\iota'(j)}$.
Up to symmetry, it is $\iota(j)<i=\iota'(j)$ and
$b_jb'_j$ avoids crossing $c^2_{i-1}c^2_{i}$.
Then $b_jb'_j$ has to cross the ``vertical'' $F_1$-edge $c^2_{i-1}c^3_{i-1}$
which costs $(k+1)(k-i+1)\,T$ in weighted crossing by (F\ref{it:thickT1}).
Recall that also $b'_{j-1}b'_j$ crosses $c^2_{i-1}c^3_{i-1}$.
Hence if we ``slide'' the vertex $b'_j$ along $b'_{j-1}b'_j$ to the face
$B'_{i-1}$, we avoid (at least) the crossing between
$c^2_{i-1}c^3_{i-1}$ and $b_jb'_j$, and replace the crossing of
$c^2_{i-1}c^3_{i-1}$ with $b'_{j-1}b'_j$ by that with $b'_{j}b'_{j+1}$.
Since the difference between the weights of $b'_{j-1}b'_j$ and
$b'_{j}b'_{j+1}$ is at most $k$ by (F\ref{it:horizk3}), 
the change in the weighted crossing number
of the whole joint embedding is $\leq k(k-i+1)\,T-(k+1)(k-i+1)\,T<-T$,
again as needed for our inductive argument.

Hence we may assume that always $\iota(j)=\iota'(j)$ but,
up to symmetry, $\iota(j)=i>j$ for some $j\in\{1,\dots,k\}$.
Analogously to the previous paragraph, we now ``slide'' each one of 
the vertices $b_j$ and $b'_j$ along the edges
$b_{j-1}b_j$ and $b'_{j-1}b'_j$ to the faces $B_{i-1}$ and $B'_{i-1}$, respectively.
The change in edge crossings is as follows:
$c^1_{i-1}c^2_{i-1}$ is newly crossed by $b_{j}b_{j+1}$ instead of former $b_{j-1}b_j$,
~$c^2_{i-1}c^3_{i-1}$ crossed by $b'_{j}b'_{j+1}$ instead of $b'_{j-1}b'_j$,
and $b_jb'_j$ crosses $c^2_{i-2}c^2_{i-1}$ instead of $c^2_{i-1}c^2_{i}$.
This by (F\ref{it:horizk3}) leads to the following change in the weighted crossing number
\begin{eqnarray*}
(i-1)T&\cdot(t_{k-j}&-t_{k+1-j}) + (k+1-i)T\cdot(t_j-t_{j-1}) +0
	\\&=& T\cdot[ (i-1)(j-k-1) + (k+1-i)j ]
	\\&=& T\cdot[(j-i+1)(k+1)-j] \leq -jT\leq-T
,\end{eqnarray*}
which is as desired.
The proof is finished.
%%%%11
% \ifaccumulating{\qed\end{proof}}
\qed\end{proof}
\end{accumulate}
% \qed\end{proof}

\ifaccumulating{\section{Reduction from Anchored Planar Crossing Number}}%
\ifnoaccumulating{\section{Reduction from Anchored Planar Crossing Number:
	 Proof of Theorem~\ref{thm:maintheorem}}}
\label{sec:reduction}
%%%%%%%%%%%%%%%%%%%%%%%%%%%%%%%%%%%%%%%%%%%%%%%%%%%%%%%%%%%
\label{sec:anchoredCM}

We prove our main theorem at the end of this section. 
The additional ingredient we need is the hardness of a special variant
of the so-called anchored crossing number problem in the plane.
In general, an {\em anchored
drawing}~\cite{DBLP:journals/siamcomp/CabelloM13} of a graph $G$ is a drawing of $G$ in
a closed disc~$D$ such that a set $A\subseteq V(G)$ of selected {\em anchor} vertices
are placed in specific points of the boundary of $D$ and the rest of the
drawing lies in the interior of~$D$.

We shall use the following very restrictive version of the problem
which we call the {\em anchored crossing number of a pair of planar graphs}:
The input is a pair of disjoint connected planar graphs $(G_1,G_2)$, their anchor sets
$A_1\subseteq V(G_1)$ and $A_2\subseteq V(G_2)$, and a cyclic permutation
$\sigma$ of $A_1\cup A_2$.
The task is to find the minimum number of crossings over all anchored
drawings of $G_1+G_2$ such that the anchors appear on the disk boundary in
the cyclic order specified by~$\sigma$.
As before, the problem is considered in the edge weighted form.

\begin{theorem}[Cabello and Mohar, \cite{DBLP:journals/siamcomp/CabelloM13}]
\label{thm:CMhard}
The anchored weighted crossing number problem of the pair of planar graphs
$(G_1,G_2)$, with anchor sets $(A_1,A_2)$ and permutation $\sigma$,
is NP-hard even under the following assumptions:
\begin{enumerate}[({A}1)]
\item\label{it:CMunique}
each of the graphs $G_1,G_2$ itself has a unique anchored embedding, and
\item\label{it:CMcut}
there is a partition $A_2=A_2^1\cup A_2^2\cup A_2^3\cup A_2^4$ such that, for
$i=1,2,3,4$, the set $A_2^i$ is consecutive in $\sigma$ restricted to $A_2$,
and the set of edges incident with $A_2^i$ forms a minimum weight cut in $G_2$
separating $A_2^i$ from $A_2\setminus A_2^i$.
\end{enumerate}
\end{theorem}

Figure~\ref{fig:CM} illustrates the hardness construction used in
\cite{DBLP:journals/siamcomp/CabelloM13}, and the conditions
(A\ref{it:CMunique}) and (A\ref{it:CMcut}) of Theorem~\ref{thm:CMhard},
which are not explicitly stated in \cite{DBLP:journals/siamcomp/CabelloM13}
but can easily be verified there.

Notice, moreover, in Figure~\ref{fig:CM} that the graph $G_2$ also has
some ``diagonal'' minimum weight cuts which use the dashed red edges of
weight only~$w-1$.
Hence, for example, every minimum weight cut of $G_2$ between $A_2^1\cup A_2^2$
and $A_2^3\cup A_2^4$ has to use some of the dashed red edges and so
cannot have all its edges incident to~$A_2$.
Consequently, the partition of $A_2$ into the four sets in (A\ref{it:CMcut})
is not just an artifact of the visual shape of $G_2$ in Figure~\ref{fig:CM}
but necessity.
% This is an artifact of the Cabello/Mohar construction, I understand, but it
% seems that really all you need is that the edges incident to $A_2$ are a min-cut

\begin{figure}[t]
$$\includegraphics[width=0.39\hsize]{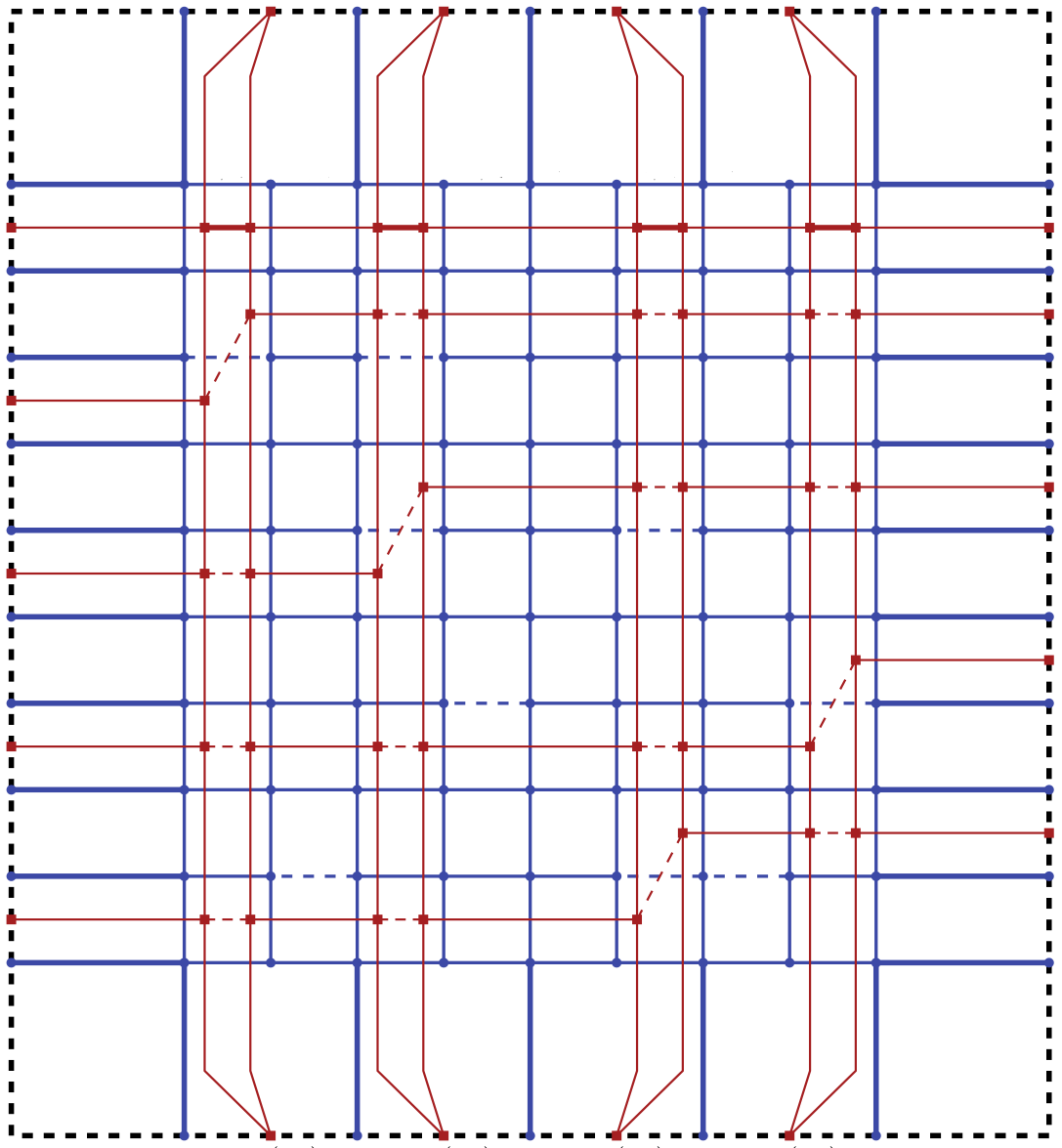}\qquad
\includegraphics[width=0.45\hsize]{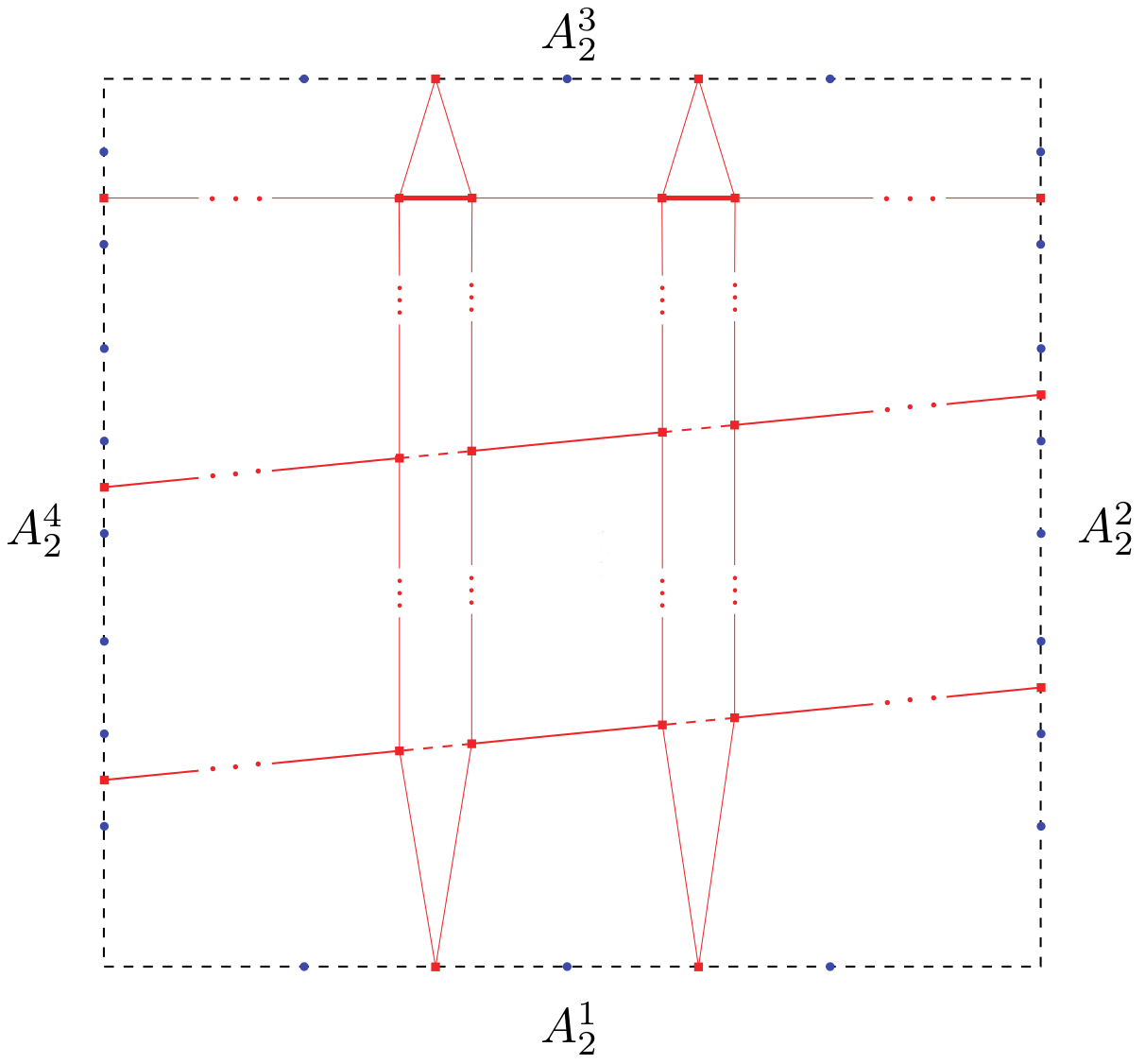}$$
\caption{An example of the construction of a hard anchored crossing number
	instance $(G_1,G_2)$ taken from \cite{DBLP:journals/siamcomp/CabelloM13}:
	$G_1$ is blue and $G_2$ is red (detailed alone on the right).
	The solid thin red edges all have weight $w$ (where $w$ is a large
	integer) and the middle dashed red edges have weight~$w-1$.}
\label{fig:CM}
\end{figure}

\medskip
We now establish the final key reduction required to prove Theorem~\ref{thm:maintheorem}.

\begin{theorem}
\label{thm:CMreduction}
There is a polynomial reduction from the special anchored crossing number problem
given in Theorem~\ref{thm:CMhard} to the 
{\sc OP-Homeo-Invariant $6$-FA Joint Planar Crossing Number} problem.
% face-anchored joint planar crossing number problem with $6$ face anchors.
\end{theorem}

\begin{proof}
We will use the instance $\ca F_{k,T}$ 
(where sufficiently large $k,T$ will be specified later)
of joint planar embedding of $(F_1,F_2)$ 
with the face anchors $\{(C_i,a_i):$ $i=1,2,3,4\}$,
from Section~\ref{sec:multiplying} in the following way.
The graph $F_1$ is joined with its mirror copy such that the anchor faces
$C_3,C_4$ get identified with $\bar{C}_3,\bar{C}_4$ 
of the copy in a ``horizontal mirror'' way, resulting in the graph $F_1^+$.
Similarly, $F_2^+$ results by joining $F_2$ with its mirror copy and
identifying $a_3,a_4$ with the copies $\bar a_3,\bar a_4$, respectively.
The resulting instance of joint planar embedding of $(F_1^+,F_2^+)$
with the six face anchors $\{(C_1,a_1),(C_2,a_2),(C_3=\bar C_3, a_3=\bar a_3),
 (C_4=\bar C_4, a_4=\bar a_4),(\bar C_1,\bar a_1),(\bar C_2,\bar a_2)\}$,
as depicted in Figure~\ref{fig:F1F2double},
will be shortly denoted by $\ca F^+$.

Let $\crgj{}(\ca F^+)$ shortly denote the (optimum) weighted crossing number
of this instance $\ca F^+$, which equals twice the value computed in
\eqref{eq:T3beta} by Lemma~\ref{lem:F1F2}.
\ifnoaccumulating{It follows, in particular, from Lemma~\ref{lem:F1F2} 
that any feasible solution of $\ca F^+$ other than the depicted one 
exceeds $\crgj{}(\ca F^+)$ by at least~$T$.}%
% 
% Moreover, one can easily argue that one of the middle face anchors,
% namely $(C_3=\bar C_3, a_3=\bar a_3)$, can be safely abandoned
% in $\ca F^+$ since the two (red) edges incident to $a_3=\bar a_3$ are
% of weight $t_0$, cf.~(F\ref{it:horizk3}) above, which is the least weight on
% the path $Q_1\cup\overline{Q_1}$.
% Hence $\ca F^+$ has the claimed property with $5$ face anchors.

Consider an instance of the anchored crossing number problem,
i.e., a pair of weighted planar graphs 
$(G_1,G_2)$ with anchor sets $(A_1,A_2)$ and permutation $\sigma$
satisfying (A\ref{it:CMunique}) and (A\ref{it:CMcut}) for the partition
$A_2=A_2^1\cup A_2^2\cup A_2^3\cup A_2^4$ in a suitable cyclic order.
% as given in Theorem~\ref{thm:CMhard},
% where the input is $(G_1,G_2)$ with anchor sets $(A_1,A_2)$ and permutation
% $\sigma$, and the partition $A_2=A_2^1\cup A_2^2\cup A_2^3\cup A_2^4$ of the
% specified properties in a suitable cyclic order.
Our aim is to construct from it an instance $\ca H$ of face-anchored joint
planar crossing number, formed by a pair of graphs $(H_1,H_2)$, such that
$H_1\supseteq F_1^+$ and $H_2\supseteq F_2^+$ and $\ca H$ inherits the six
face anchors of~$\ca F^+$.
Furthermore, we will show with the help of (A\ref{it:CMunique}) that $\ca H$ is
orientation-preserving homeo-invariant.
% \medskip

% \setcounter{topnumber}{1}% to force this picture on the next page
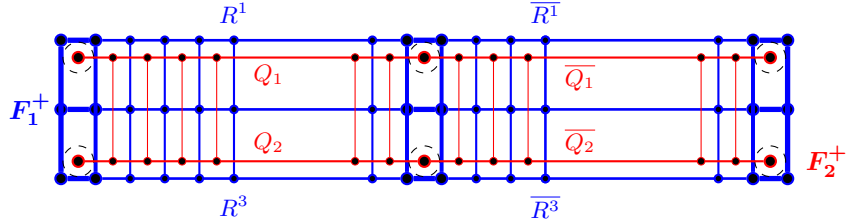
\begin{figure}[t]
\begin{center}%\vspace*{-2ex}
\begin{tikzpicture}[scale=0.23]
\tikzstyle{every node}=[draw, shape=circle, minimum size=2pt,inner sep=1.5pt, fill=black]
\tikzstyle{every path}=[thick, color=blue]
\draw (0,0) node (x1) {}; \draw (2,0) node (x2) {};
\draw (0,4) node[label=left:\mbox{\normalsize\boldmath$F_1^+\!\!$}] (y1) {};
\draw (2,4) node (y2) {};
\draw (0,8) node (z1) {}; \draw (2,8) node (z2) {};
\draw (20,0) node (x3) {}; \draw (22,0) node (x4) {};
\draw (20,4) node (y3) {}; \draw (22,4) node (y4) {};
\draw (20,8) node (z3) {}; \draw (22,8) node (z4) {};
\tikzstyle{every path}=[line width=2pt, color=blue]
\draw (x2) -- (x1) -- (z1) -- (z2);
\draw (x3) -- (x4); \draw (z4) -- (z3);
\draw (y1) -- (y2); \draw (y3) -- (y4);
\tikzstyle{every path}=[line width=1.5pt, color=blue]
\draw (x2) -- (z2); \draw (x3) -- (z3);
\draw (x4) -- (z4);
\tikzstyle{every path}=[line width=1.0pt, color=blue]
\draw (x2) -- (x3); \draw (y2) -- (y3); \draw (z2) -- (z3);
\tikzstyle{every node}=[draw, shape=circle, minimum size=1pt,inner sep=1pt, fill=black]
\tikzstyle{every path}=[thick, color=blue]
\draw (4,0) node {} -- (4,4) node {} -- (4,8) node {};
\draw (6,0) node {} -- (6,4) node {} -- (6,8) node {};
\draw (8,0) node {} -- (8,4) node {} -- (8,8) node {};
\draw (10,0) node[label=below:$R^3$] {} -- (10,4) node {}
	-- (10,8) node[label=above:$R^1$] {};
\draw (18,0) node {} -- (18,4) node {} -- (18,8) node {};
\tikzstyle{every path}=[thin, color=black]
\draw (1,7) node[dashed,fill=none, minimum size=3ex] {};
\draw (1,1) node[dashed,fill=none, minimum size=3ex] {};
\draw (21,7) node[dashed,fill=none, minimum size=3ex] {};
\draw (21,1) node[dashed,fill=none, minimum size=3ex] {};
\tikzstyle{every node}=[draw, shape=circle, minimum size=2pt,inner sep=1.5pt, fill=black]
\tikzstyle{every path}=[thick, color=red]
\draw (1,1) node (a1) {};
\draw (1,7) node (a2) {};
\draw (21,1) node (a3) {};
\draw (21,7) node (a4) {};
\draw (a1) -- (a3); \draw (a2) -- (a4);
\tikzstyle{every node}=[draw, shape=circle, minimum size=1pt,inner sep=1pt, fill=black]
\tikzstyle{every path}=[thin, color=red]
\draw (3,1) node {} -- (3,7) node {};
\draw (5,1) node {} -- (5,7) node {};
\draw (7,1) node {} -- (7,7) node {};
\draw (9,1) node {} -- (9,7) node {};
\draw (17,1) node {} -- (17,7) node {};
\draw (19,1) node {} -- (19,7) node {};
\draw (12,0.5) node[draw=none, fill=none, label=above:$Q_2$] {};
\draw (12,7.5) node[draw=none, fill=none, label=below:$Q_1$] {};
\tikzstyle{every node}=[draw, shape=circle, minimum size=2pt,inner sep=1.5pt, fill=black]
\tikzstyle{every path}=[thick, color=blue]
\draw (40,0) node (x5) {}; \draw (42,0) node (x6) {};
\draw (40,4) node (y5) {}; \draw (42,4) node (y6) {};
\draw (40,8) node (z5) {}; \draw (42,8) node (z6) {};
\tikzstyle{every path}=[line width=2pt, color=blue]
\draw (x5) -- (x6) -- (z6) -- (z5);
\draw (y5) -- (y6);
\tikzstyle{every path}=[line width=1.5pt, color=blue]
\draw (x5) -- (z5);
\tikzstyle{every path}=[line width=1.0pt, color=blue]
\draw (x4) -- (x5); \draw (y4) -- (y5); \draw (z4) -- (z5);
\tikzstyle{every node}=[draw, shape=circle, minimum size=1pt,inner sep=1pt, fill=black]
\tikzstyle{every path}=[thick, color=blue]
\draw (24,0) node {} -- (24,4) node {} -- (24,8) node {};
\draw (26,0) node {} -- (26,4) node {} -- (26,8) node {};
\draw (28,0) node {} -- (28,4) node {} -- (28,8) node {};
\draw (28,0) node[label=below:$\overline{R^3}$] {} -- (28,4) node {}
	-- (28,8) node[label=above:$\overline{R^1}$] {};
\draw (38,0) node {} -- (38,4) node {} -- (38,8) node {};
\tikzstyle{every path}=[thin, color=black]
\draw (41,7) node[dashed,fill=none, minimum size=3ex] {};
\draw (41,1) node[dashed,fill=none, minimum size=3ex] {};
\tikzstyle{every node}=[draw, shape=circle, minimum size=2pt,inner sep=1.5pt, fill=black]
\tikzstyle{every path}=[thick, color=red]
\draw (41,1) node[label=right:\mbox{\normalsize\boldmath$~~F_2^+\!\!$}] (a5) {};
\draw (41,7) node (a6) {};
\draw (a3) -- (a5); \draw (a4) -- (a6);
\tikzstyle{every node}=[draw, shape=circle, minimum size=1pt,inner sep=1pt, fill=black]
\tikzstyle{every path}=[thin, color=red]
\draw (23,1) node {} -- (23,7) node {};
\draw (25,1) node {} -- (25,7) node {};
\draw (27,1) node {} -- (27,7) node {};
\draw (37,1) node {} -- (37,7) node {};
\draw (39,1) node {} -- (39,7) node {};
\draw (30,0.5) node[draw=none, fill=none, label=above:$\overline{Q_2}$] {};
\draw (30,7.5) node[draw=none, fill=none, label=below:$\overline{Q_1}$] {};
\end{tikzpicture}
% \vspace*{-4ex}
\end{center}
\caption{A ``join'' of the instance $\ca F_{k,T}$ from
	Figure~\ref{fig:F1F2joint} and of its horizontal mirror copy, giving a
	planar joint embedding instance $\ca F^+$ with $6$ face anchors.}
	%%(where one shared anchor of the two copies is abandoned),
	%and the unique depicted optimal solution.}
\label{fig:F1F2double}
\end{figure}

Recall the horizontal paths $Q_i$, $i=1,2$, in $F_2$ connecting $a_i$ to $a_{i+2}$,
and the horizontal paths $R^j$, $j=1,3$, in $F_1$ connecting $x^j_2$ to $x^j_3$.
Let $\overline{Q_i}$ and $\overline{R^j}$ denote their mirror copies in $\ca F^+$.
For sufficiently large $k$, we can easily construct injective mappings
$\alpha:A_1\to V(R^1\cup R^3\cup\overline{R^3}\cup\overline{R^1})$ and
$\beta_1:A_2^1\to V(Q_1)\setminus\{a_1,a_3\}$,
$\beta_2:A_2^2\to V(Q_2)\setminus\{a_2,a_4\}$,
$\beta_3:A_2^3\to V(\overline{Q_2})\setminus\{\bar a_1,a_3\}$,
$\beta_4:A_2^4\to V(\overline{Q_1})\setminus\{\bar a_2,a_4\}$,
such that the images of $A_1\cup A_2$ under the respective mappings, when
pictured in Figure~\ref{fig:F1F2double}, occur exactly in the cyclic order
specified by~$\sigma$.

Let $\beta=\beta_1\cup\beta_2\cup\beta_3\cup\beta_4$.
We define the graph $H_1$ from a disjoint union of $F_1^+$ and $G_1$,
by identifying the vertex $x$ with $\alpha(x)$ for each $x\in A_1$.
Similarly, we define $H_2$ as $F_2^+\cup G_2$ after identifying $y$ with
$\beta(y)$ for each $y\in A_2$.
The homeo-invariant property of $\ca H$ will easily follow from
Lemma~\ref{lem:F1F2} and property (A\ref{it:CMunique}) for the following pair of
embeddings $(H_1',H_2')$:
for $i=1,2$, $H_i'$ is the unique plane embedding of $H_i$ such that the
restriction of $H_i'$ to $F_i^+$ is as in Figure~\ref{fig:F1F2double}.
% Let $\ca H$ shortly denote the derived face-anchored joint planar embedding instance of
% $(H_1,H_2)$ (with the same set of six face anchors as $\ca F^+$).

Let the weighted anchored crossing number of $(G_1,G_2)$ with
anchor sets $(A_1,A_2)$ and cyclic permutation~$\sigma$ equal~$s$.
We assume that $T=\Omega(k^6)$ is chosen sufficiently large such
that~$T>s$.
For $i=1,2,3,4$, let $w_i$ be the minimum weight of a cut in $G_2$
separating $A_2^i$ from $A_2\setminus A_2^i$;
by (A\ref{it:CMcut}), $w_i$ equals the sum of weights of the edges incident to~$A_2^i$.
Then, there is a drawing $H'$ of $H_1+H_2$ 
with $\crgj{}(\ca F^+)+(w_1+w_2+w_3+w_4)\cdot T^2+s$ weighted crossings,
where the term $(w_1+w_2+w_3+w_4)\cdot T^2$ accounts for crossings between
the $G_2$-edges incident with $A_2$ and the edges of 
$R^1\cup R^3\cup\overline{R^3}\cup\overline{R^1}$,
such that $H'$ is a joint orientation-preserving homeomorphic embedding
of $(H_1',H_2')$.

We finish the proof by showing that if the (weighted) face-anchored joint
crossing number of $\ca H$ equals~$r$,
then there exists an anchored drawing of $(G_1,G_2)$ respecting $(A_1,A_2)$ and~$\sigma$,
with at most $r':=r-(w_1+w_2+w_3+w_4)\cdot T^2-\crgj{}(\ca F^+)$ crossings.
This will then automatically imply that the aforementioned drawing $H'$
which is joint orientation-preserving homeomorphic to $(H_1',H_2')$,
is also an optimal solution of~$\ca H$.
\ifaccumulating{The details are again left for the full
preprint~\cite{arxivJoint}.}
\begin{accumulate}
\ifaccumulating{\subsection{Proof of Theorem~\ref{thm:CMreduction};
	technical details}
	\begin{proof}[continued]}
We may assume that $k$ is sufficiently large such
that~$(w_1+w_2+w_3+w_4)=o(k^3)$.
Take any drawing $H^0$ of $H_1+H_2$ which is an optimal joint embedding
solution to $\ca H$, i.e. such that $\crgj{}(H^0)=r$.
When $H^0$ is restricted to $F_1^++F_2^+$, the number of crossings is at least
$\crgj{}(\ca F^+)$ by definition.
$F_1^+$ has a unique plane embedding as in Figure~\ref{fig:F1F2double},
and if any of the $F_2^+$-paths $Q_1\cup\overline{Q_1}$,
$Q_2\cup\overline{Q_2}$ entered the $F_1^+$-face incident with
$R_1\cup\overline{R_1}\cup R_3\cup\overline{R_3}$,
then already $H^0$ restricted to $F_1^++F_2^+$ would have more than
$\crgj{}(\ca F^+)+\Omega(k^3)\cdot T^2$ crossings by Lemma~\ref{lem:F1F2}.
The latter contradicts optimality of $H^0$ since there exists
a feasible solution with $\crgj{}(\ca F^+)+(w_1+w_2+w_3+w_4)\cdot T^2+s\leq
 \crgj{}(\ca F^+)+o(k^3)\cdot T^2+T$ crossings,
as shown above.

Let $F^0_j$ denote the restriction of $H^0$ to $F_j^+$, 
and let $G^0_j$ be the restriction of $H^0$ to $G_j$, for $j=1,2$.
Note that $G^0_2$ is drawn in the outer face of $F^0_2$.
Consequently, any $G_2$-path from $A_2^1$ to $A_2\setminus A_2^1$
either has to cross the $F_1^+$-path $R^1$ (of weight $T^2$), 
or else it has to to make at least $T^3$ weighted crossings with one of the anchor faces
of $G^0_1$.
However, the latter cannot happen since $\crgj{}(\ca F^+)+T^3$ crossings
is more than the optimum $r\leq\crgj{}(\ca F^+)+o(k^3)\cdot T^2$.
Analogous claims hold for $A_2^2,A_2^3,A_2^4$ and
$R_3,\overline{R_3},\overline{R_1}$, respectively.

By the minimum-cut property of the instance $(G_1,G_2)$, as formulated in
Theorem~\ref{thm:CMhard}, we hence account for at least
$\crgj{}(\ca F^+)+(w_1+w_2+w_3+w_4)\cdot T^2= r-r'$ 
weighted crossings which involve edges of $F_1^+$.
The number of crossings in the restriction of $H^0$ to $G_1+G_2$ thus is at
most~$r'$, as desired.
However, we still have to prove that this restriction is an anchored drawing
of $G_1+G_2$.

Since $r'<T$ by our assumption, there can be no more crossings of $G^0_2$
with $F^0_1$ than those with $R_1\cup\overline{R_1}\cup
 R_3\cup\overline{R_3}$ accounted for above.
Hence we can draw in $H^0$ a simple curve $\gamma_1$ starting in $x_2^1$ and passing
through $b_1,c_1^1,b_2,c_2^1,b_3,\dots,c_{k-1}^1,b_k$ to $x_3^1$,
such that $\gamma_1$ is disjoint from $H^0$ except at the listed vertices.
Analogous curves $\gamma_2,\bar\gamma_2,\bar\gamma_1$ can be drawn
alongside $Q_2,\overline{Q_2},\overline{Q_1}$, and these together
with some edges of $F_1^0$ form a simple closed curve
which plays the role of the disk boundary in an anchored drawing 
of $G_1+G_2$ restricted from~$H^0$.
\ifaccumulating{\qed\end{proof}}
\end{accumulate}
\qed\end{proof}

\ifnoaccumulating{\medskip}
\begin{proof}[of Theorem~\ref{thm:maintheorem}]
Theorem~\ref{thm:maintheorem} for the {\sc Joint Crossing Number} problem
and genus $6$ follows imediately by the chain of
reductions from Theorem~\ref{thm:face-anchored},
and Theorems~\ref{thm:CMhard} and~\ref{thm:CMreduction}. 
\ifnoaccumulating{In the case of simple $3$-connected graphs on the input 
we additionally employ Proposition~\ref{pro:face-anchored-3}. }%
For genus greater than $6$, it suffices to add dummy face anchors 
in the reduction of Theorem~\ref{thm:face-anchored}.
\ifnoaccumulating{\par}%
Finally, for hardness of the {\sc Homeomorphic} and {\sc OP-Homeomorphic} variants,
we can simply use the same reductions---by the orientation-preserving
homeo-invariant promise, the (hard) instances produced by the chain of
reductions have the same solution value in all the three problem variants.
\qed\end{proof}

\section{Conclusions}\label{sec:conclusions}
%%%%%%%%%%%%%%%%%%%%%%%%%%%%%%%%%%%%%%%%%%%%%%%%%%%%%%%%%%%

The following is another immediate consequence of
Theorems~\ref{thm:face-anchored} and~\ref{thm:CMhard}:

%Moreover, the following follows from the first two of the reductions: WHICH>>>>>

\begin{theorem}\label{thm:facetheorem}
The {\sc $h$-FA Joint Planar Crossing Number} problem
is NP-hard for every~$h\geq6$.
\end{theorem}

There is yet another interesting consequence.
The main result of aforementioned~\cite{DBLP:journals/siamcomp/CabelloM13}
is that {\sc Crossing Number} is NP-hard even on {\em almost-planar
graphs}, i.e.\ those which can be made planar by removing one edge.
Their hardness reduction, derived from hard anchored crossing number
instances as shown in Figure~\ref{fig:CM},
essentially uses an unbounded number of vertices of arbitrarily high degrees.
Elaborating on the reduction of our proof of Theorem~\ref{thm:CMreduction}, while
using a special gadget derived from $\ca F^+$ turned inside out, 
we can give back the following strengthening%
\ifaccumulating{ (we omit the proof due to space constraints)}:

\begin{theorem}[slight improvement
    upon~\cite{DBLP:journals/siamcomp/CabelloM13}]
\label{thm:stren}%
The {\sc Crossing Number} problem remains NP-hard
even if the input is restricted to almost-planar graphs having 
a bounded number, namely at most~$16$, vertices of degree greater than~$3$.
\end{theorem}

Note that, on the other hand, Cabello and Mohar
\cite{DBLP:journals/algorithmica/CabelloM11} prove that
{\sc Crossing Number} is solvable in linear
time if the input is an almost-planar graph with all vertices except
for the
two of the planarizing edge having degree at most~$3$.

Another natural extension of our results would be to prove
Theorem~\ref{thm:maintheorem} for non-orientable surfaces.
This is not inherently difficult---it suffices to replace the toroidal gadgets $T_i$
(cf.~Figure~\ref{fig:Tfaceanchors}) with suitable projective grids,
and to use a crosscap instead of each toroidal handle.
However, a formal statement would require us to repeat most of the arguments
of the proof of Theorem~\ref{thm:face-anchored},
and hence we refrain from giving the full statement in this
short paper.

A question worth further investigation is how small the genus in
Theorem~\ref{thm:maintheorem} and the number of face anchors in
Theorem~\ref{thm:facetheorem} can be for the statements to hold.
\ifnoaccumulating{%
Recent improvements in our reductions, related to Theorem~\ref{thm:stren},
suggest that perhaps $6$ can be replaced by $4$ in these theorems.
}

\begin{small}
\bibliographystyle{plain}
\bibliography{phcross}
\end{small}

\end{document}
%%%%%%%%%%%%%%%%%%%%%%%%%%%%%%%%%%%%%%%

\ifaccumulating{%
\newpage\appendix\sloppy
\section{Appendix}
%%%%%%%%%%%%%%%%%%%%%%%%%%%%%%%%%%%%%%%
\def\thesubsection{A-\arabic{subsection}}
\accuprint
}

\end{document}